\documentclass[11pt]{article}
\usepackage{geometry}
\usepackage{authblk}
\usepackage{amsmath,amssymb,mathtools,bm}
\usepackage{amsthm}
\let\originalproof\proof
\let\endoriginalproof\endproof
\RenewDocumentEnvironment{proof}{o}{%
  \IfNoValueTF{#1}{\originalproof}{\originalproof[\proofname\ (#1)]}%
}{\endoriginalproof}
\usepackage{booktabs,array,graphicx,xurl,placeins}
\usepackage[round]{natbib}
\usepackage{microtype}
\usepackage{tikz}
\usepackage{algorithm,algorithmicx,algpseudocode,dsfont}
\usetikzlibrary{positioning,arrows.meta,calc,fit}
\newtheorem{theorem}{Theorem}
\newtheorem{proposition}{Proposition}
\newtheorem{lemma}{Lemma}
\newtheorem{corollary}{Corollary}
\newtheorem{definition}{Definition}
\newtheorem{assumption}{Assumption}

\theoremstyle{remark}
\newtheorem{example}{Example}

\newwrite\figuredescriptions
\immediate\openout\figuredescriptions=DescriptionTexts.txt
\newcounter{figuredescription}
\newcommand{\Description}[1]{%
  \stepcounter{figuredescription}%
  \begingroup\stepcounter{figure}%
  \immediate\write\figuredescriptions{\arabic{figuredescription} (Fig. \thefigure) - \detokenize{#1}}%
  \endgroup\addtocounter{figure}{-1}%
}
\usepackage{hyperref}
\hypersetup{pdftitle={Effective-Balance Design in Ethereum Proof-of-Stake},
  pdfauthor={Bolin Zhang and Xuechao Wang}}
\newcommand{\keywords}[1]{\par\smallskip{\small\noindent\textbf{Keywords:} #1\par}}
\begin{document}

\title{Effective-Balance Design in Ethereum Proof-of-Stake}
\date{}
\author[1]{Bolin Zhang}
\author[1]{Xuechao Wang\thanks{Corresponding author: \texttt{xuechaowang@hkust-gz.edu.cn}.}}
\affil[1]{HKUST(GZ)}

\maketitle
\setcounter{footnote}{0}

\begin{abstract}
Proof-of-Stake (PoS) systems offer a scalable and energy-efficient alternative to Proof-of-Work (PoW). The economic security of PoS relies critically on incentivizing broad and meaningful participation through well-cali\-brated staking reward mechanisms.  We study Ethe\-reum's effective-balance rules as a mechanism-design problem in which the protocol chooses a reward coefficient, an activation threshold, and a maximum effective balance, and validators respond in the induced propor\-tional-allocation game.

For fixed participating validators and identity counts, we establish a unique Nash equilibrium in validator-level total stakes. When validators can split or merge their visible identities, each candidate identity structure must be evaluated after the market re-clears. Restructuring incentives therefore depend on post-reset payoff comparisons that account for all validators' stake responses.  For an illustrative simulation, we construct address-level  budget proxies from Beacon Chain deposits and withdrawals, Lido staking flows, Etherscan balances, and calibrated synthetic agents. With illustrative cost coefficients, lowering the reserve price from 32 ETH to 30 ETH increases active identities by 5.922\% and total effective stake by 0.098\%.

\end{abstract}

\section{Introduction}

The cornerstone of permissionless blockchains is a resource-based consensus protocol, whose security fundamentally relies on the amount of resources actively committed to the network. Bitcoin pioneered this idea with proof-of-work (PoW), and its implied energy wastage spurred the development of alternative designs such as proof-of-stake (PoS). Notably, Ethereum—one of the most widely used platforms—transitioned from PoW to PoS in 2022, marking a major milestone in blockchain design \citep{buterin2022proof} aimed at improving scalability and reshaping the economic underpinnings of network security.  

PoS and PoW differ in how participants commit resources to secure consensus~\citep{john2025proof}. In Ethereum, validators are selected to propose blocks with probabilities proportional to their effective stake and receive rewards for attestation duties. Honest participation earns staking rewards, while missed duties can incur penalties and specified conflicting proposals or votes can trigger slashing, potentially forfeiting the validator's stake~\citep{ethereumdocs}. This design exemplifies \textit{cryptoeconomic security}: an attacker must acquire and risk losing a substantial share of the total stake to compromise the system, making large-scale attacks prohibitively expensive. As more ETH is staked, the cost of a majority attack  rises and its feasibility wanes.

The staking reward mechanism is crucial for driving active participation: if rewards are too low, validators may opt out, weakening the network; if they are too high, they can encourage concentration or unsustainable token issuance. Ethereum's dynamic reward curve reduces marginal returns as total stake grows \citep{ethereumdocs}, but participation also depends on effective-balance constraints.

Ethereum's Pectra upgrade changed the effective-balance rules \footnote{https://ethereum.org/en/roadmap/pectra/}. Before Pectra,  each validator identity had a maximum effective balance of 32 ETH~\citep{ethereumdocs}. EIP-7251 raises this maximum effective balance to 2,048 ETH for compounding validators while keeping the 32 ETH activation threshold.\footnote{\url{https://ethereum.org/en/roadmap/pectra/}; \url{https://eips.ethereum.org/EIPS/eip-7251}.} The change lets large operators consolidate stake into fewer validators, but it leaves the entry threshold for smaller stakers unchanged. 
Under our expected-reward abstraction, we use two properties of Ethereum’s reward mechanism: aggregate issuance scales with the square root of total effective stake, and each validator’s expected protocol reward is proportional to its effective stake.

The central question of the paper is how reserve and cap rules move total effective stake and the operational visible-identity count when validators can also alter their visible identities.
We use these security--participation coordinates to compare the outcomes induced by Ethereum's effective-balance rule. The coordinates of this frontier are the total effective stake $S$, which raises the economic cost of an attack, and the operational visible-identity count $D_{\mathrm{op}}$, which is an observable indicator of identity-level operational decentralization.
The latter captures the extensive margin of protocol-facing participation.
The security–participation frontier depends on how validators adjust their stakes and identities in response to protocol rules. We therefore analyze both the protocol’s parameter choices and the resulting validator equilibrium.

\paragraph{Our contribution.} 
We study Ethereum staking as a mechanism-design problem with an induced proportional-allocation game. The protocol first commits to the reward coefficient, the activation threshold, and the maximum effective balance. Validators then respond by choosing their identities and effective stake. This gives the model a two-level structure: the protocol stage is a design commitment, and the validator stage is the induced equilibrium response.

The analysis considers two settings under complete information among validators. 
First, fixing the participating validators and one identity per validator, we prove that the induced proportional-allocation game has a unique Nash equilibrium. Second, we allow each validator to change its identity count (by splitting or merging) while coordinating the stake behind all its identities. In this game, each count deviation is evaluated after all validators re-optimize stake. Then, every feasible count profile has unique continuation totals and payoffs for the validators, and the induced finite restructuring game admits a mixed equilibrium.

We connect the model to data by constructing address-level bid limits from observed staking activity and available ETH balances, then use  these on-chain budget proxies in an illustrative reserve-price  simulation around the current 32 ETH threshold. Under the assumed costs, reducing the reserve from 32 ETH to 30 ETH increases active identities by 5.922\% and total effective stake by 0.098\%.

In summary, this paper makes the following contributions:
\begin{itemize}
    \item We propose a two-level mechanism-design framework that models Ethereum staking as a proportional-allocation game and treats PoW as a proportional-selection boundary case.
    
    \item We establish a unique fixed-participation quantity equilibrium, identify cap-induced splitting incentives, and formulate the protocol designer's optimization problem.
    
    \item We  provide an illustrative simulation of the model with blockchain-derived budget proxies  to examine how reserve  prices affect total effective stake and visible identity counts under  assumed costs.
\end{itemize}

\noindent \textbf{Paper structure.} Section~\ref{Preliminary} defines the Ethereum PoS abstraction and the expected-reward interpretation. Section~\ref{non-separable} characterizes the static induced game. Section~\ref{dynamic section} studies splitting and merging of validator identities under reserve and cap constraints. Section~\ref{simulation} explains how data are converted into bid limits and reports reserve-price counterfactuals. Section \ref{sec:discussion} concludes.

\subsection{Related Work}  

Classical auction and mechanism-design models evaluate allocation rules using objectives such as revenue and efficiency~\citep{myerson1981optimal,vickrey1961counterspeculation,christodoulou2015efficiency}. Related work also studies fairness, participation, and externalities~\citep{feldman1974fairness,seddighin2019externalities}. In our Ethereum staking model, the designer values total effective stake and operational participation.
These coordinates describe the resources committed to consensus and the breadth of active protocol-facing participation.

Proportional-allocation mechanisms are studied in profit sharing, budget-feasible procurement, contests, and shared-resource allocation~\citep{bachrach2013incentives,amanatidis2023partial,birmpas2022parallel,kelly1997charging,feldman2008proportional}. A central assumption in allocation-rate and scalar-parameterized mechanism results is that utility from the received allocation is strictly increasing and concave~\citep{johari2004efficiency,johari2009efficiency}. Tullock contests use similar proportional-selection logic to link costly investment to winning probabilities~\citep{garimidi2025tullock}. In our staking game, however, rewards depend on both the allocation share and total stake, so this fixed utility representation does not apply directly.
We therefore adapt the divisible-resource auction framework of~\citet{maheswaran2003nash} to characterize Nash equilibria with endogenous staking rewards.
Appendix~\ref{appendix:tullock-benchmark} develops the comparison.

Prior blockchain literature demonstrates that reward mechanisms inherently shape stake concentration. For example, geometric rewards can mitigate compounding in PoS wealth dynamics~\citep{fanti2019compounding}. Additionally, cap-and-margin reward sharing can steer stake pools toward a target structure~\citep{brunjes2020reward}, while redistribution rules generalize proportional allocation in PoW and PoS mining pools~\citep{CAN2022274}. In a different vein, \citet{chen2023gacha} models mining process as a single-agent Markov decision process. In contrast, we explicitly model the strategic interactions among validators to analyze equilibrium staking participation under Ethereum's specific protocol targets within a game-theoretic framework.

\section{Preliminaries and Model} \label{Preliminary}
\subsection{Preliminaries}
This section defines the institutional objects that enter the model and fixes the distinction between economic agents and protocol-facing validator identities.

\paragraph{Ethereum PoS rewards.} In each epoch, Ethereum assigns proposer and attestor duties to active validator identities. A realized slot has a single proposer and a committee of attestors, while the ex-ante probability of receiving these duties is increasing in effective stake. Active validators earn protocol rewards for correct participation, incur penalties for missed duties, and can be slashed for specified conflicting proposals or votes~\citep{ethereumdocs,chen2024game}. The payoff below is therefore an expected protocol-reward benchmark: a risk-neutral validator evaluates rewards over proposer and attestor randomness, and the expected reward share is proportional to effective stake.

We exclude additional reward sources such as EIP-1559 transaction tips \citep{roughgarden2021transaction} and maximal extractable value (MEV) \citep{daian2020flash,garimidi2025tullock}. These revenues arise from transaction ordering and execution-layer activity rather than from the staking reward formula, so they are not directly controlled by the protocol here. We also abstract from adversarial consensus strategies \citep{zhang2024max}, network latency, and restaking because they introduce attack timing, propagation, and cross-protocol risk channels outside the effective-balance mechanism. 

Prior to Pectra, each validator identity had an effective balance capped at 32 ETH. EIP-7251 raises the maximum effective balance to 2,048 ETH while keeping the 32 ETH activation threshold. In the notation below, the activation threshold is a reserve price \(\underline r\), and the maximum effective balance is an upper bound \(\bar r\). The base-reward formula can be written schematically as
\begin{equation}
    \text{base reward} = \frac{\textit{effective balance} *\texttt{BASE\_REWARD\_FACTOR}}{\sqrt{\text{total balance}}*\texttt{BASE\_REWARDS\_PER\_EPOCH}}.
    \label{base reward}
\end{equation}

\paragraph{Validators and identities.} A \emph{validator} is an economic decision maker, such as a solo staker, staking pool, exchange, or liquid-staking operator, that may operate multiple validator identities.
A \emph{validator identity} is the protocol-facing key/effective-balance unit that enters the selection rule.
The protocol applies balance constraints to each identity separately.
Thus the  participation measure below is operational participation at the identity level.

\subsection{Model}
Ethereum staking can be modeled as an auction over a divisible reward pool. Validators commit ETH to earn protocol rewards from proposing blocks and performing attestation duties. Although individual block proposals are discrete, averaging over protocol randomness gives a divisible allocation of expected rewards. Under the expected-reward abstraction used here, each identity receives a share of aggregate rewards proportional to its effective stake. Stakes therefore act as bids in a \emph{proportional auction}~\citep{kelly1997charging}, while validators incur a cost of committing funds. Furthermore, the protocol imposes a minimum staking threshold of 32 ETH, which functions as a ``reserve price'' that determines eligibility for participation. This auction-theoretic abstraction allows us to analyze staking incentives using established tools from mechanism design and resource allocation theory.

Let $\mathbf s=(s_{i\ell})$ denote the effective-stake bids of the active validator identities, where $i\in\{1,\ldots,M\}$ indexes validators and $\ell$ indexes their identities.
Validator $i$ has staking budget $B_i \in (0, \infty)$ and operates $L_i$ identities with total $Q_i=\sum_{\ell=1}^{L_i}s_{i\ell}\le B_i$.
Write $\mathbf Q=(Q_1,\ldots,Q_M)$ and
\[
 Q_{-i}=\sum_{j\ne i}Q_j,\qquad Q=\sum_{j=1}^{M}Q_j=Q_i+Q_{-i}.
\]

In the static one-identity setting, $Q_i=s_i$, while the dynamic setting allows several identities per validator.
Identity-control relationships, budgets, cost parameters, and the protocol design are common knowledge among validators.
Economic validators act noncooperatively: deviations are unilateral, and cross-validator coalitions are excluded.

\begin{assumption}[Validator-share admissibility]\label{assumption 1}
For each active economic validator $i$, the total stake share satisfies
\[
 \frac{Q_i}{Q}
 =\frac{\sum_{\ell=1}^{L_i}s_{i\ell}}{\sum_{j=1}^{M}Q_j}
 <\frac23.
\]
\end{assumption}

Assumption~\ref{assumption 1} delimits the decentralized protocol interpretation of an outcome.
If one economic validator controls at least two-thirds of total effective stake across its identities, it can unilaterally finalize its preferred valid chain without votes from other validators~\citep{buterin2020combining,ethereumdocs}. 
We treat this regime as a breakdown of the decentralized system modeled here, in which the proportional-allocation game and its utility specification no longer apply. 

The rest of the model describes the induced staking interaction.
The protocol design fixes eligibility and reward rules, while validators jointly choose their identities' effective stakes subject to the identity bounds and validator budgets.
We formalize this induced interaction as a bidding game.

\begin{definition}[Bidding game]\label{def:bidding-game}
For a fixed protocol design $\mathcal D$ and feasible identity counts $(L_i)_{i=1}^M$, each active validator $i$ jointly chooses its identities' stakes subject to its budget and the per-identity balance limits.
The induced outcome is
\[
\left(\{B_i\}_{i=1}^{M},\{Q_i\}_{i=1}^{M},M,\mathcal D\right).
\]
\end{definition}

\subsubsection{Designer Objective}

The protocol designer selects design parameters \(\mathcal D\) to move the observable security--participation frontier. The two coordinates are total effective stake, which can raise the economic cost of corruption, and the operational visible-identity count, which records the participation margin available to the protocol.

\begin{definition}[Security]\label{security}
   The security coordinate, denoted by \(S\), is the total effective stake supplied by visible identities:
   \(
   S = \sum_{i=1}^{M} Q_i.
   \)
\end{definition}

\begin{definition}[Operational participation] \label{decentralization}
   The operational visible-identity count, denoted by \(D_{\mathrm{op}}\), is the number of active protocol-facing identities after any identity restructuring and stake re-equilibration:
   \(
   D_{\mathrm{op}} = \sum_{i=1}^{M} L_i.
   \)
   Identities controlled by the same economic validator are counted separately.
\end{definition}

The designer chooses protocol parameters to maximize $\Phi(S,D_{\mathrm{op}})$, which is increasing in both arguments and evaluated at the induced validator equilibrium.
Here, $D_{\mathrm{op}}$ captures the extensive margin of protocol-facing participation and it indicates operational decentralization without identifying independent ownership or Sybil resistance.

We focus on visible identities because the protocol records their keys and effective balances and applies the reserve and cap at this level~\citep{eip7251}, whereas common economic control need not be identifiable from on-chain records.
In the dynamic game (Section~\ref{dynamic section}), the population of potential economic validators is fixed, but entry and exit can change its active subset.
Counting active economic validators would capture these entry and exit decisions but would miss identity splits and merges within an unchanged active set.
The count $D_{\mathrm{op}}$ captures both types of adjustment after stake re-equilibration.
It therefore represents the identity-level dimension of operational decentralization studied here: the breadth of active protocol-facing participation.

\subsubsection{Validator Payoff} 

Under an expected-reward interpretation, we aggregate protocol rewards over a fixed accounting period. To isolate the economic incentives of staking, we focus on validators that jointly choose their identities' effective-stake bids while abstracting away complex adversarial vectors such as selfish mining, latency exploitation, eclipse attacks, and restaking mechanics.

Although validator duties recur, we model stake choice for a fixed identity structure as a \textbf{one-shot game}.
Validators simultaneously commit stake for a representative period with fixed protocol parameters, budgets, and cost schedules.
Taking expectations over proposer and attestor randomness gives a single payoff for each stake profile.
Thus, the repeated process can be effectively analyzed as a single-stage optimization.

Formally, validator $i$ chooses its total stake $Q_i$, taking opponents' total $Q_{-i}$ as given in a unilateral quantity decision.
Let $C_i(Q_i)$ denote its cost over the same period.
With $Q=Q_i+Q_{-i}>0$ and $\alpha>0$, the square-root issuance rule gives
\begin{equation}\label{utility1}
 V_i(Q_i;Q_{-i})=
 \underbrace{\alpha\sqrt Q}_{\text{aggregate reward}}
 \underbrace{\frac{Q_i}{Q}}_{\text{reward share}}
 -C_i(Q_i)
 =\frac{\alpha Q_i}{\sqrt{Q_i+Q_{-i}}}-C_i(Q_i).
\end{equation}

We specify the validator cost as
\begin{equation}\label{eq:owner_cost}
 C_i(Q_i)=c_iQ_i+\frac{\kappa_i}{2}Q_i^2,
 \qquad c_i>0,\quad\kappa_i\ge0.
\end{equation}

The linear term captures foregone returns or incremental expenses, rather than a loss of stake principal.
The quadratic term captures increasing funding costs as staking reduces the liquid buffer available during withdrawal delays.\footnote{\url{https://ethereum.org/staking/withdrawals/}.}
For example, let $W_i>0$ with $W_i\ge B_i$ be liquid resources and let a funding need be uniform on $[0,W_i]$.
Committing $Q_i$ leaves a liquid buffer of $W_i-Q_i$.
A funding need above this buffer creates a shortfall that must be financed or leaves a payment unmet.

For $Q_i>0$, the probability of a shortfall is $Q_i/W_i$, and its average size conditional on occurring is $Q_i/2$.
With loss $h_i\ge0$ per unit of shortfall, the expected funding cost is
\[
\begin{aligned}
 \mathcal C_i(Q_i)
 &=h_i\,\frac{Q_i}{W_i}\,\frac{Q_i}{2}
 =\frac{h_i}{2W_i}Q_i^2,
 \qquad \kappa_i=\frac{h_i}{W_i}.
\end{aligned}
\]

Thus a larger commitment increases both the probability and the size of a funding shortfall.
This gives the quadratic term in Equation~\eqref{eq:owner_cost} and links its coefficient to liquidity resources and the unit shortfall loss.
All identities operated by validator $i$ share the same resources and funding need, so their cost depends on the total $Q_i=\sum_\ell s_{i\ell}$.
Splitting an unchanged total across more identities therefore preserves the funding cost.
Appendix~\ref{appendix:cost_foundation} gives the general convex-cost derivation and the joint-cost expansion across identities.

An inactive validator has $L_i=Q_i=0$ and payoff zero, including at the all-inactive profile.
For mathematical characterization, the formula specifies payoffs on all resource-feasible profiles.
Section~\ref{non-separable} restricts validators to one identity, while Section~\ref{dynamic section} changes identity counts and re-solves the same game.

\section{Static Setting}\label{non-separable}

This section’s analysis extends the one-shot auction framework of~\citet{maheswaran2003nash} to model Ethereum’s PoS staking as a static game in which a fixed set of active validators simultaneously selects its bids.
We begin by formulating each agent’s optimization problem under reserve price and bid limit constraints.
Using first- and second-order conditions, we derive each validator’s optimal bidding strategy and prove existence and uniqueness of a Nash equilibrium as \citet{maheswaran2003nash}.

\subsection {Constrained Optimization for Validators} 

We introduce several constraints to formalize the model. 
In the static setting, each validator operates at most one identity, with $Q_i=s_i$ when active.
The assignment of identities to validators, budgets, costs, and protocol parameters are common knowledge before bidding.
Validator $i$ has budget $B_i$.
An agent is \emph{active} if it submits a valid bid $Q_i\in[\underline r,B_i]$.
Let $E(\underline r)=\{i:B_i\ge \underline r\}$ be the budget-eligible set and fix a nonempty participating subset $A\subseteq E(\underline r)$.
Validators outside $A$ supply zero in this conditional game, with entry and exit reserved for Section~\ref{dynamic section}.
For each active validator $i\in A$, the static problem is:

\begin{equation}\label{optimziation}
 \max_{Q_i\in[\underline r,B_i]} V_i(Q_i;Q_{-i})
 =\frac{\alpha Q_i}{\sqrt{Q_i+Q_{-i}}}-c_i Q_i-\frac{\kappa_i}{2}Q_i^2,
\end{equation}
where $Q_{-i} = \sum_{j \neq i} Q_j$ denotes the aggregate bids from all agents other than agent $i$.

The first-order condition for an interior optimum is
\begin{equation}
    \frac{\partial V_i}{\partial Q_i}
    =\frac{\alpha(Q_{-i}+Q_i/2)}{(Q_i+Q_{-i})^{3/2}}-c_i-\kappa_i Q_i,\label{first order}
\end{equation}
with second derivative
\begin{equation}
    \frac{\partial^2 V_i}{\partial Q_i^2}
    =-\frac{\alpha(Q_{-i}+Q_i/4)}{(Q_i+Q_{-i})^{5/2}}-\kappa_i<0.\label{second order}
\end{equation}

For the active eligible set \(A\), Each participating owner must also satisfy IR condition:
\begin{equation} \label{eq:static_IR}
 V_i(Q_i;Q_{-i})\ge0
 \quad\Longleftrightarrow\quad
 \frac{\alpha}{\sqrt Q}\ge c_i+\frac{\kappa_i}{2}Q_i,
 \qquad i\in A.
\end{equation}

\begin{proposition}\label{prop opt}
For the optimization problem in Equation~\eqref{optimziation}, the constrained best response is
\begin{equation} \label{eq:static_optimal_bid}
    Q_i^{\mathrm{BR}}(Q_{-i})=\max\{\underline r,\min\{Q_i^\dagger,B_i\}\},
\end{equation}
where $Q_i^\dagger$ is the unique maximizer of $V_i(Q_i;Q_{-i})$ over $Q_i\ge0$.
\end{proposition}

The strict concavity of the payoff makes the constrained optimum the projection of its unconstrained maximizer onto the feasible interval.
Appendix~\ref{section proof 1} gives the KKT conditions and the participation check.

The proposition has a direct design interpretation.
Lowering \(\underline r\) can add agents to the eligible set and can also reduce the minimum stake required of agents who were previously active.
Raising a validator's budget \(B_i\) weakly expands its feasible set, but it need not increase the chosen stake once the unconstrained optimum is already feasible.
The empirical budget distribution therefore matters for eligibility and for the amount of stake validators can supply.
 Validators with budgets below the reserve cannot  participate.
Those with budgets far above the cap may require identity restructuring before their additional capacity can enter effective stake.
Section~\ref{dynamic section} examines how identity restructuring changes the feasible stake choices.

\subsection{Equilibrium Characterization}\label{sec:price-demand}
Proposition~\ref{prop opt} solves one identity's constrained choice for fixed opponents' aggregate stake.
A Nash equilibrium requires these individual best responses to be mutually consistent.
Following the proportional-allocation representation in~\citet{maheswaran2003nash}, we express this consistency requirement through a candidate aggregate stake $Q>0$.
Write $x_i=Q_i/Q$ for the stake share of validator $i$'s single identity and define the unit reward $p(Q)=\alpha/\sqrt Q$.
Using $Q=Q_i+Q_{-i}$, an interior optimum satisfies
\[
 p(Q)\left(1-\frac{x_i}{2}\right)=c_i+\kappa_i Q x_i.
\]

The factor $1-x_i/2$ retains the validator's effect on aggregate rewards.
Solving for the share and applying the bid constraints gives the following representation.

\begin{proposition}[Demand representation]\label{prop:price_demand}
For each active owner $i\in A$ and $Q>0$, the constrained demand share is
\begin{equation}\label{eq:static_demand}
 d_i(Q)=\frac{1}{Q}\operatorname{proj}_{[\underline r,B_i]}
 \left(Q\frac{p(Q)-c_i}{p(Q)/2+\kappa_i Q}\right),
\end{equation}
where $\operatorname{proj}_{[a,b]}(y)=\max\{a,\min\{y,b\}\}$.
Each $d_i$ is continuous and nonincreasing.
A feasible bid vector is a conditional quantity equilibrium exactly when $Q_i=Q d_i(Q)$ and $\sum_{i\in A}d_i(Q)=1$.
\end{proposition}

The derivation and monotonicity proof are in Appendix~\ref{section proof 2}.

This demand representation follows from the strategic optimality conditions, rather than an additional behavioral assumption.
For any candidate total $Q$, the projection enforces the reserve and bid limit on the implied quantity, and division by $Q$ gives its share.
Let aggregate share demand be $D_A(Q)=\sum_{i\in A}d_i(Q)$.
Market clearing requires $D_A(Q)=1$, so the resulting bids add up to the same total used to compute them.

\begin{theorem}[Existence and uniqueness of static equilibrium]\label{thm: unique NE}
For every fixed protocol design with nonempty active set \(A\), there exists a unique aggregate stake $Q_A>0$ satisfying \(D_A(Q_A)=1\). The fixed-identity proportional-allocation game then has a unique Nash equilibrium.
\end{theorem}

\begin{proof}[Sketch]
Aggregate demand is continuous and brackets one at the sums of the lower and upper quantity bounds.
At a clearing total, increasing the proposed aggregate strictly lowers every validator's demand share, so no second clearing total exists.
The demand representation then gives a unique quantity vector.
Appendix~\ref{section proof theorem unique} gives the full argument for arbitrary positive quantity intervals.\qedhere
\end{proof}

The existence and uniqueness of the Nash equilibrium established above provide a solid foundation for analyzing strategic behavior in the static setting. Since each agent has a well-defined best response and the resulting strategy profile is stable and predictable, we can treat different initial representations that lead to the same equilibrium outcome as equivalent:

\begin{lemma}\label{Equivalence of bidding game}
The two bidding games
\[
\begin{gathered}
 \left(\{B_i\}_{i=1}^M,\{Q_i^*\}_{i=1}^M,M,\mathcal D\right)\quad\text{and}\\
 \left(\{Q_i^*\}_{i=1}^M,\{Q_i^*\}_{i=1}^M,M,\mathcal D\right)
\end{gathered}
\]
are equivalent as they yield the same result.
\end{lemma}

This lemma serves strictly as an outcome-equivalence statement for a fixed set of identities. It does not imply that unused bid capacity becomes irrelevant once identity restructuring is permitted. In fact, Section~\ref{dynamic section} investigates exactly how this unused capacity matters when an economic agent can add or merge visible identities.

In practice, an iteration adapted from~\citet{maheswaran2003nash} can seek the conditional Nash equilibrium. Appendix~\ref{appendix:static algorithm} gives the update equation and the corresponding algorithm. The same proportional-allocation backbone also gives a PoW boundary case: if hash power replaces effective stake and the block reward is fixed, the payoff retains a proportional resource-share structure.
See Appendix~\ref{pow section} for more details.

\section{Dynamic Setting}\label{dynamic section}

The static analysis fixes participation and the visible identity structure.
We now allow validators to choose how many identities to operate.
Each identity change is evaluated after all validators re-optimize their stakes.
We first use a symmetric example without a cap to show that a split can increase total stake and the active identity count yet reduce the splitting validator's payoff.
This conflict between private incentives and the protocol's objectives motivates our analysis of a per-identity effective-balance cap.
We examine when a cap makes the split profitable and improves both protocol coordinates.
We then formulate the designer's optimization over the resulting equilibria.

\subsection{Identity Restructuring}\label{subsec:reset}
The first task is to make the identity problem visible.
We first study identity restructuring under the uncapped design $\mathcal D=(\alpha,\underline r)$.
The assignment of identities to validators is common knowledge among validators, but the protocol applies the activation threshold to each visible identity.
The validator compares post-reset payoffs, while the protocol cares about the resulting pair $(S,D_{\mathrm{op}})$.
This possible mismatch motivates our analysis of the upper-bound rule below in Section \ref{solution section}.

Let $\mathcal E=\{1,\ldots,M\}$ denote potential validators, with budgets $B_i$ and feasible identity counts
\begin{equation}\label{eq:count_menu}
 L_i\in\mathcal L_i(\underline r)=\{0,1,\ldots,\lfloor B_i/\underline r\rfloor\}.
\end{equation}

The choice $L_i=0$ gives inactivity and zero payoff.
For this fixed design, the validator response has two stages: validators first simultaneously choose identity counts, then observe $\bm L$ and simultaneously choose stakes conditional on those counts.
Each active validator jointly chooses its identities' stakes to maximize $V_i(Q_i;Q_{-i})$, subject to $L_i\underline r\le Q_i\le B_i$.
Splitting a fixed total among identities leaves the validator's reward and cost unchanged.
However, each additional identity raises its minimum total stake, so a split can change the quantity equilibrium.
Theorem~\ref{thm: unique NE} extends to these validator-specific intervals and yields a unique quantity equilibrium $\mathbf Q^*(\bm L)$ for every nonzero feasible count profile.
\subsubsection{Continuous Allocation and the Quantity Reset}\label{appendix:quantity_reset}
\begin{lemma}[Reduction to validator quantities]\label{lem:reduction}
For each feasible positive count $k\in\mathcal L_i(\underline r)$, define
\[
 b_i^k=\begin{cases}
 B_i, & \text{without a cap},\\
 \min\{B_i,k\bar r\}, & \text{with a cap}.
 \end{cases}
\]

For a validator operating $L_i>0$ identities, the feasible totals are exactly $a_i=L_i\underline r\le Q_i\le b_i^{L_i}$.
Every allocation of the same total has the same payoff, and equal allocation $s_{i\ell}=Q_i/L_i$ is feasible.
The continuation game is therefore
\[
 \max_{Q_i\in[a_i,b_i^{L_i}]}
 \left\{\frac{\alpha Q_i}{\sqrt{Q_i+Q_{-i}}}-C_i(Q_i)\right\}.
\]
\end{lemma}

\begin{corollary}[Single-valued validator quantity reset]\label{thm:reset}
Every feasible non\-zero count profile induces a unique equilibrium in validator totals and payoffs.
Define, for a feasible positive count $L$ and $z>0$,
\begin{equation}\label{eq:quadratic_reset}
 Q_{i,L}(z)=\operatorname{proj}_{[L\underline r,\,b_i^L]}
 \frac{z[p(z)-c_i]}{p(z)/2+\kappa_i z}.
\end{equation}

The equilibrium total is the unique solution of $\sum_{i:L_i>0}Q_{i,L_i}(z)=z$ between the sums of the lower and upper bounds.
Active quantities are $Q_i^*(\bm L)=Q_{i,L_i}(z)$.
The all-zero profile has zero quantities and payoffs.
\end{corollary}

Appendix~\ref{appendix:quantity_reset_proof} proves the reduction and quantity-reset statements.

Write $Q_{-i}^*(\bm L)=\sum_{j\ne i}Q_j^*(\bm L)$.
Writing $V_i(\bm L)$ for the game payoff evaluated at this continuation, the payoff and gain from changing the count to $k$ are
\begin{align}
 V_i(\bm L)&=V_i\!\left(Q_i^*(\bm L);Q_{-i}^*(\bm L)\right),\label{eq:validator_post_reset_value}\\
 \Delta_i(k\mid\bm L)&=V_i(k,\bm L_{-i})-V_i(\bm L).\notag
\end{align}

\begin{definition}[Restructuring game]\label{def:restructuring-game}\label{def:split_merge}
For a fixed design $\mathcal D$, the restructuring game has economic validators $i\in\mathcal E$ as players, feasible identity-count sets $\mathcal L_i(\underline r)$ as action sets, and continuation payoffs $V_i(\bm L)$ from Equation~\eqref{eq:validator_post_reset_value}.
It is the finite first-stage game obtained by evaluating each count profile at its ensuing quantity equilibrium.
A feasible split, merge, entry, or exit from $L_i$ to $k$ is strictly profitable if $\Delta_i(k\mid\bm L)>0$.
\end{definition}

Each comparison holds other validators' identity counts fixed, while all validators re-optimize their stakes.
Equilibrium requires that no validator can gain from any feasible identity-count change, including entry and exit.
Appendix~\ref{appendix:dynamic algorithm} describes the numerical procedure.

\subsection{Symmetric Equilibria}\label{subsec:symmetric_participation}

With identical costs and constraints, a fixed identity structure yields a common equilibrium stake.
Fix $n\ge1$ active validators with common parameters $(B,c,\kappa)$, each operating a feasible positive count $1\le\ell\le\lfloor B/\underline r\rfloor$.
They face the same quantity constraints and cost $C(q)=cq+\kappa q^2/2$.

\begin{proposition}[Symmetric quantity characterization]\label{prop:symmetric}
Let $a=\ell \underline r$, $b=B$, and
\begin{equation}\label{eq:symmetric_marginal}
 F_{n,\ell}(q)=\frac{\alpha}{\sqrt{nq}}\left(1-\frac{1}{2n}\right)-c-\kappa q.
\end{equation}

The common quantity $q_{n,\ell}$ equals $a$ if $a=b$ or $F_{n,\ell}(a)\le0$, and equals $b$ if $F_{n,\ell}(b)\ge0$.
Otherwise it is the unique interior zero of $F_{n,\ell}$.
The continuation has total stake $Q=nq_{n,\ell}$, $D_{\mathrm{op}}=n\ell$, and validator payoff $V_{n,\ell}=p(nq_{n,\ell})q_{n,\ell}-C(q_{n,\ell})$.
\end{proposition}

Symmetry and quantity uniqueness imply a common stake, while the decreasing marginal payoff determines the active bound or interior root.
The proof is in Appendix~\ref{appendix:symmetric_proofs}.

\paragraph{Unilateral restructuring.}
The symmetric equilibrium provides the baseline payoff for a unilateral identity-count deviation.
A feasible deviation to $k>0$ identities leaves the other $n-1$ validators symmetric, although their stakes adjust after the reset.
For a candidate aggregate stake $z>0$, the demand representation in Section~\ref{sec:price-demand} gives
\[
 q_k(z)=\operatorname{proj}_{[k\underline r,\,B]}
 \left(\frac{z[p(z)-c]}{p(z)/2+\kappa z}\right),
\]
where $p(z)=\alpha/\sqrt z$.
If the deviator changes from $\ell$ to $k$ identities, the new total stake solves
\begin{equation}\label{eq:asymmetric_reset}
 q_k(z)+(n-1)q_\ell(z)=z,
\end{equation}
and the deviator compares $V_i(q_k(z);(n-1)q_\ell(z))$ with its original payoff $V_{n,\ell}$.
\begin{example}[Uncapped split incentives]\label{ex:uncapped_restructuring}
Consider the symmetric benchmark in which each validator initially operates one identity.
For $n\ge2$ validators, let $q_n^0$ be the common equilibrium stake.
If $B\ge2\underline r$, a validator that splits into two identities changes its payoff from $u_n^0$ to $v_n^0$, where
\[
 u_n^0=\alpha\sqrt{\frac{q_n^0}{n}}-C(q_n^0),
 \qquad
 v_n^0=\frac{\alpha q_2(z_n)}{\sqrt{z_n}}-C(q_2(z_n)).
\]

Here $z_n$ is the post-split total determined by Equation~\eqref{eq:asymmetric_reset} with $\ell=1$ and $k=2$.
In particular, for $\alpha=13/6$, $c=1$, $\kappa=3/20$, $B=2$, and $\underline r=5/8$, we have $v_2^0<u_2^0$ and $v_3^0>u_3^0$, which means such splitting is profitable for three validators but not two.
\end{example}

\subsection{Upper Bound and the Private Split Incentives}\label{solution section}

For $n=2$, Example~\ref{ex:uncapped_restructuring} shows that a split can improve $(S,D_{\mathrm{op}})$ yet remain privately unprofitable.
To address this incentive conflict, we introduce a finite per-identity effective-balance cap $\bar r\ge\underline r$, giving $L_i\underline r\le Q_i\le\min\{B_i,L_i\bar r\}$.
Proposition~\ref{prop:symmetric} and Equation~\eqref{eq:asymmetric_reset} apply with these bounds for all validators.
The cap is nonbinding when $\bar r\ge\max_i B_i$.

A binding cap lowers the value of staying with the old visible identity list because excess capacity behind one identity is no longer credited. Splitting can then route that capacity through another reserve-clearing identity. Thus a cap is useful here only if it both makes the validator adopt the split and improves the resulting \((S,D_{\mathrm{op}})\) coordinates.

\begin{example}[Cap-induced splitting]\label{ex:reset_split_flip}
Retain the parameters of Example~\ref{ex:uncapped_restructuring} and denote the corresponding capped payoffs by $u_n^{\mathrm{cap}}$ and $v_n^{\mathrm{cap}}$.
For $n=2$ and $\bar r=0.8$, the pre- and post-split quantity equilibria are $(0.8,0.8)$ and $(5/4,0.8)$, giving $v_2^{\mathrm{cap}}>u_2^{\mathrm{cap}},(S,D_{\mathrm{op}})=\left(\frac{41}{20},3\right).$ The split is therefore adopted and improves both coordinates over the uncapped outcome.
\end{example}

\begin{figure}[!htbp]
\centering
\includegraphics[width=.92\linewidth]{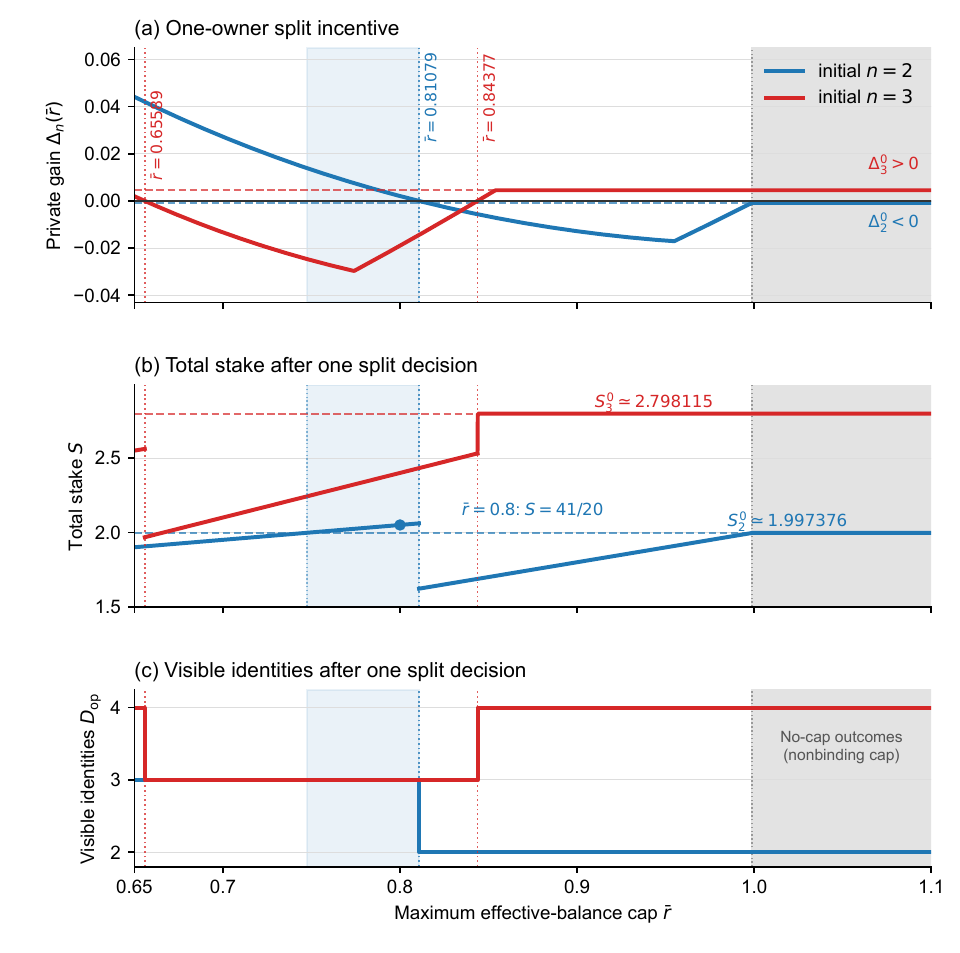}
\Description{Three panels show the gain from one validator splitting, the resulting total stake, and the visible identity count for two and three initially symmetric economic validators across effective-balance caps.}
\caption{One-validator split comparisons for initially symmetric $n=2$ (blue) and $n=3$ (red) cases.
Panel (a) gives the post-reset payoff gain, and panels (b) and (c) give outcomes after that split decision, with no split at indifference.
Dashed horizontal lines mark the corresponding uncapped benchmarks.
Gray shading marks caps that are nonbinding in both cases, yielding the no-cap outcomes.
The blue shaded interval improves both protocol coordinates for $n=2$.
Each curve evaluates one split, and some tight-cap outcomes admit further restructuring.}
\label{fig:split_incentive_stability}
\end{figure}

Figure~\ref{fig:split_incentive_stability} separates private adoption from protocol improvement.
For $n=2$, panel (a) gives profitable splits below around $0.81$, while panel (b) requires $S=5/4+\bar r>S_2^0$ to improve total stake.
Together, these conditions identify caps that induce this split and improve both protocol coordinates.
A tighter cap can still induce splitting but reduce total stake.
For $n=3$, intermediate caps instead suppress an already profitable uncapped split.
Appendix~\ref{appendix:examples} gives the calculations.

The role of this simple example is therefore conceptual as much as numerical.
By limiting how much stake one visible identity can carry, the cap changes the validator's feasible choices across identity counts and its incentive to split, making it an integral part of protocol design.
The next subsection evaluates protocol designs through the equilibria they induce.

\subsection{Post-Reset Frontier and Protocol Design}\label{subsec:equilibrium_protocol_frontier}
The preceding subsections describe how validators restructure their identities and how the market resets after each change.
We now evaluate these responses through the protocol objective $\Phi(S,D_{\mathrm{op}})$, using observable total stake $S=Q$ and visible participation $D_{\mathrm{op}}=\sum_iL_i$.

For a fixed design $\mathcal D$, validators choose identity counts from their feasible menus, anticipating the resulting stake equilibrium.
Their payoffs at that equilibrium determine the incentives to split, merge, enter, or exit.

\begin{definition}[Post-reset  equilibrium]\label{def:separation_combination}
A feasible profile is a pure post-reset equilibrium if
\begin{equation}\label{eq:count_NE}
 V_i(\bm L)\ge V_i(k,\bm L_{-i})
 \quad\text{for every }i\text{ and }k\in\mathcal L_i(\underline r).
\end{equation}
\end{definition}

\begin{theorem}[Post-reset equilibrium and individual rationality]\label{thm:post_reset_partition_equilibrium}
The restructuring game admits a mixed Nash equilibrium that extends to a subgame-perfect equilibrium.
Each validator's payoff is nonnegative at pure equilibria and nonnegative in expectation at mixed equilibria.
\end{theorem}

\paragraph{The Designer's Problem.}
The designer chooses $\mathcal D=(\alpha,\underline r,\bar r)$ to maximize its own objective $\Phi(S,D_{\mathrm{op}})$ through validators' equilibrium response.
For a fixed population, let $\Theta$ be a finite nonempty design menu and $\mathcal N(\mathcal D)$ the mixed-equilibrium set of the restructuring game.
Given an equilibrium-selection rule $\eta$ with $\eta(\mathcal D)\in\mathcal N(\mathcal D)$, the designer solves
\begin{equation}\label{eq:protocol_problem}
 \max_{\mathcal D\in\Theta}
 \mathbb E_{\eta(\mathcal D)}[\Phi(S,D_{\mathrm{op}})].
\end{equation}

The selection rule describes which equilibrium follows a design and is held fixed when comparing designs.
It is not an additional action controlled by the designer.
For a selected pure equilibrium, the objective is simply $\Phi(S(\mathcal D),D_{\mathrm{op}}(\mathcal D))$.
Without a specified selection rule, equilibrium outcomes form a correspondence rather than a single-valued response.

A finite design menu guarantees attainment once preferences and equilibrium selection are specified.
Among selected pure outcomes, a maximizer of strictly increasing $\Phi$ is Pareto efficient within that selected set.
Appendix~\ref{section proof protocol frontier selection} records these observations.

\section[Illustrative Simulation]{Illustrative Simulation} \label{simulation}
Building on the preceding analysis, we combine  on-chain budget proxies with a dynamic bidding algorithm to  illustrate staking responses to protocol parameters, particularly the reserve price.
We construct an address-level sample from Beacon Chain stakers, Lido participants, and synthetic records generated from the Lido bid-limit distribution with multiplicative noise. 
Each distinct address or synthetic record supplies a budget proxy for a potential validator.
Lido participants supply funds through a pooled staking service, and an economic owner can hold multiple addresses.
The simulation maps each budget proxy to one decision maker.
We then simulate owners' stake and identity choices under the complete-information benchmark.
\subsection{Budget Proxies and Sample Construction}\label{appendix:empirical-construction}
The data construction maps on-chain address records to staking-budget proxies for the model's potential validators.
We consider three factors in data selection and usage.

The first consideration concerns the choice of participants. In alignment with the spirit of blockchain for decentralization, our focus is on solo stakers—those who individually control their staking keys. This is because true decentralization is better captured by the number of independent (solo) stakers, rather than by the total number of validators. Addresses controlled by centralized entities reduce decentralization and thus need to be excluded from our analysis. Each retained address or synthetic record represents a potential validator in the simulation.

The second consideration in the time interval. Because we mainly focus on the PoS, we start from block \texttt{15537394}, the first block following Ethereum’s transition to PoS, and end at block \texttt{22164400}, the most recent block available at the time of completing this research.

The third consideration is the data source. A natural starting point is the Beacon Chain, where all validators—whether acting independently or on behalf of organizations—stake directly to the Ethereum deposit contract at \path{0x00000000219ab540356cBB839Cbe05303d7705Fa}. Although Ethereum requires a minimum of 32 ETH to activate a validator, some transactions, such as \path{0xff0552f53e5fcdbbe84814ad207b6a3ec213c4c6bfe094e23694aad008f67f91}, deposit less than 32 ETH without validator activation. However, relying solely on Beacon Chain data renders the concept of a reserve price ineffective for deposits under 32 ETH, since these do not correspond to active validators. To address this gap, we incorporate data from Lido, a major liquid staking protocol. Lido pools ETH from users—many of whom deposit less than 32 ETH—and stakes it collectively. In return, users receive stETH tokens at a 1:1 ratio with ETH, which eliminates price distortion between different staking forms. Additionally, we collect balance data for the selected addresses directly from Etherscan. The sum of each address's net staked amount and its corresponding wallet balance is used to define the ``bid limit'' for each agent in our simulation.
Lido users supply funds to a pool whose node operators run validators~\citep{lidoStakingDocs}.
For the reserve-price comparison, each user address is assigned a direct-staking decision with its constructed budget.
Observed wallet balances and net staking amounts determine this budget proxy.
Assumption~\ref{assumption 1} is evaluated on the resulting simulated validator stakes.

Several different methodologies have been applied to identify candidate solo-staker addresses.
The dataset in~\citet{ethsolo2024} describes probabilistic machine-learning models and a later heuristic filter covering active validator identities as of June 30, 2024.
The approach in~\citet{stakecat_solo_2024} excludes identified entities and uses contract-interaction history to distinguish candidate lists.
The candidate lists use different selection criteria and observation dates.
Nevertheless, the general data acquisition strategies employed in these works are instructive.
Broadly speaking, we proceed in three steps as follows:

\textbf{Step 1: Beacon Chain data.} We collect data from the Beacon Chain using Dune\footnote{\url{https://dune.com/}}, following the methodology of \citet{stakecat_solo_2024}. Specifically, we intersect the identified addresses with Dune’s \texttt{labels.addresses} table, which tags blockchain addresses with descriptive labels. In particular, we focus on addresses labeled under the category \texttt{infrastructure}, which includes \texttt{Ethereum Validator} and \texttt{Ethereum Miner}. For each selected address, we retrieve balance data from Etherscan and compute the corresponding bid limit on the Beacon Chain.

\textbf{Step 2: Lido data.} In contrast to Beacon Chain deposits, Lido’s staking involves fewer addresses that serve multiple roles. This makes it difficult to rely solely on \texttt{labels.addresses}, which generally assigns addresses a single dominant function (e.g., \texttt{dex}). Therefore, we again follow the methodology of \citet{stakecat_solo_2024}  to identify relevant Lido addresses. After fetching balance data, we compute bid limits for these addresses. The resulting dataset is then uploaded to Dune for further classification.  As an address-selection heuristic, we filter based on the following categories: \texttt{dex}, \texttt{tornado\_cash}, \texttt{op\_governance}, \texttt{op\_traders}, and \texttt{infrastructure} (which includes both \texttt{Ethereum Validator} and \texttt{Ethereum Miner}).

\textbf{Step 3: Synthetic data generation.} To simulate additional staking behavior from entities such as CEXs, staking pools, and other liquid staking protocols (excluding Lido), we assume that the distribution of bid limits in these groups follows that of Lido. Based on this assumption, we  set the number of synthetic agents using the relation\footnote{\url{https://dune.com/hildobby/eth2-staking}}:
\begin{equation}
    G = N_{\text{Lido}} \times 
    \frac{S_{\text{non-Lido}}}{S_{\text{Lido}}},
\end{equation}
where \( G \) is the number of generations, \( N_{\text{Lido}} \) is the number of Lido users, \( S_{\text{non-Lido}} \) is the staking amount in liquid staking, CEX, and staking pools (excluding Lido), and \( S_{\text{Lido}} \) is the staking amount via Lido.

To introduce  illustrative variability into the synthetic bid limits while preserving their scale, we apply multiplicative Gaussian noise. Specifically, each bid limit \( b \), sampled from the Lido distribution, is perturbed as
\begin{equation}
    b_{\text{generated}} = max\{0,b \cdot (1 + \varepsilon)\}, \quad \varepsilon \sim \mathcal{N}(0, \sigma^2),
\end{equation}
where \( \varepsilon \) is drawn from a normal distribution with mean zero and standard deviation \( \sigma = 0.005 \). This ensures proportional fluctuations. Negative values are clipped to zero to avoid invalid bid limits.

Finally, drawing on insights from \citet{ethsolo2024}, which notes that solo stakers generally operate fewer than 200 validators, we filter our data to retain only those addresses meeting this condition. We provide a summary table detailing the address-level bid limits obtained from the Beacon Chain, Lido, and synthetically generated datasets as follows:
\begin{table}[h!]
\centering
\caption{Descriptive Statistics of Bid Limits}
\label{table description}
\begin{tabular}{l r  r r r r r r r}
\toprule
Source & Count & Mean & Min & 25\% & 50\% & 75\% & Max \\
\midrule
Beacon    & 4017   & 239.75  & 1.00   & 32.00  & 64.03  & 173.43 & 5696.00 \\
Lido      & 134133 & 7.51    & 0.10   & 0.27   & 0.78   & 2.98   & 5374.01 \\
Generated & 221267 & 9.06    & 0.10   & 0.49   & 1.00   & 4.23   & 5397.65 \\
\bottomrule
\end{tabular}
\end{table}

\subsection{Parameters and Cost Assumptions} 
Both panels use $\alpha=100$ and $\bar r=2048$ ETH, with $\underline r$ varying from 1 to 40 ETH.
$c_i=0.1$ for $B_i\ge32$ ETH and $c_i=0.105$ otherwise.
Addresses represented in the Beacon data have $\kappa_i=6\times10^{-5}$, while all others have $\kappa_i=3\times10^{-4}$.

Because \(\bar r\) is high relative to most observed bid limits, the reserve price \(\underline r\) has the most visible effect on participation in this counterfactual.

\paragraph{Parameter rationale.}
The two binary classifications distinguish the level $c_i$ and slope $\kappa_i$ of marginal staking cost, $C_i'(Q)=c_i+\kappa_iQ$.
Assigning lower $c_i$ to larger budgets represents the hypothesis that committing funds carries a smaller baseline opportunity cost for validators with greater resources.
This interpretation is motivated by the role of wealth and transaction needs in the staking--liquidity trade-off studied by \citet{cong2025tokenomics}.
The 32 ETH cutoff aligns the budget classification with Ethereum's activation requirement and remains fixed as the reserve varies.

The curvature classification distinguishes addresses represented in the Beacon data from Lido addresses and their synthetic extension, all within Ethereum.
Lido accepts small deposits and issues transferable staking claims, whereas direct staking entails validator operation and delayed access to committed funds~\citep{lidoStakingDocs}.
These differences motivate allowing preferences and constraints to vary across staking channels, as in \citet{beccuti2025stakingmarket}.
The lower $\kappa_i$ assigned to Beacon addresses represents a selection hypothesis: validators already using direct staking may face a smaller incremental liquidity burden from additional commitments.
As in Appendix~\ref{appendix:cost_foundation}, this coefficient measures the cost of committing the validator's funds, rather than the operating fees of a staking service.
The synthetic records inherit the non-Beacon assignment because they extend the Lido budget distribution.

\subsection{Numerical Results}

\begin{figure}[htbp]
\centering
\begin{minipage}[b]{.49\linewidth}
\centering
\includegraphics[height=.61\linewidth]{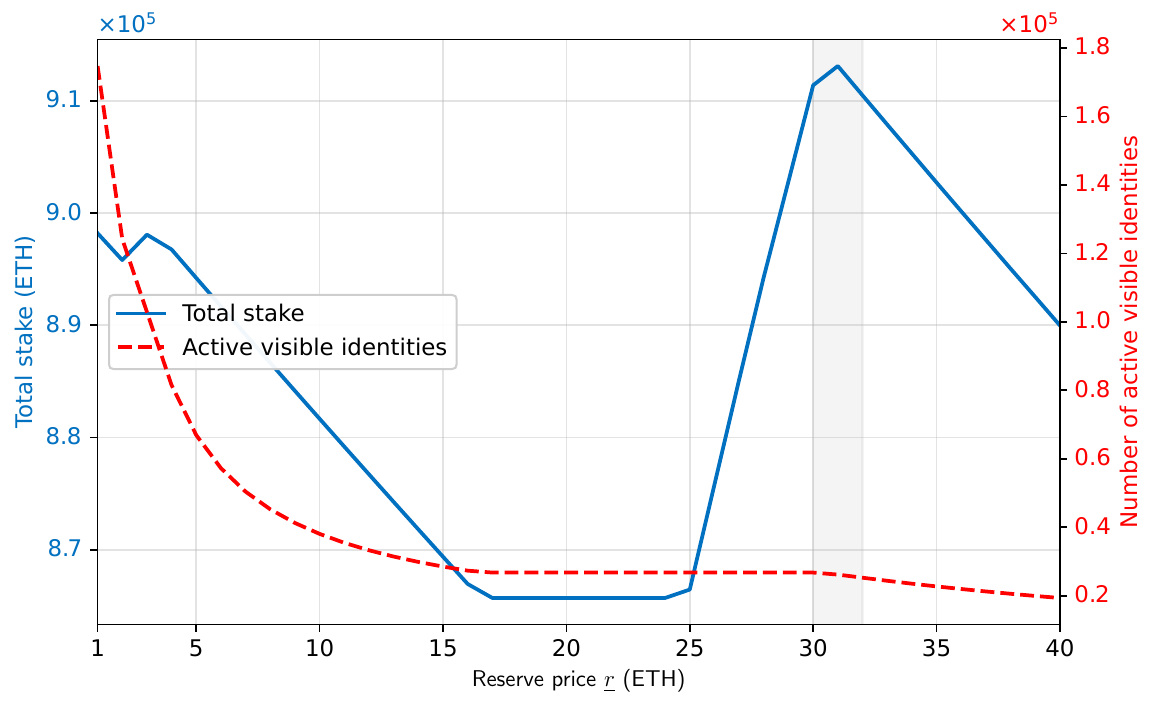}
\par (a) Reserve-price counterfactuals.
\end{minipage}\hfill
\begin{minipage}[b]{.49\linewidth}
\centering
\includegraphics[height=.61\linewidth]{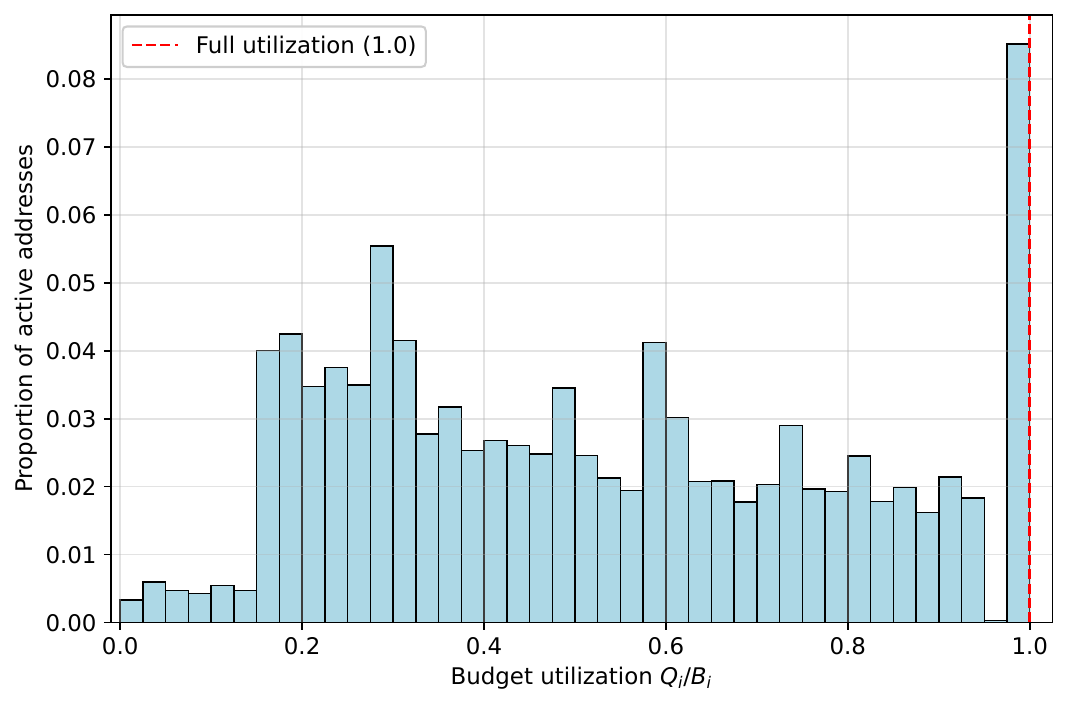}
\par (b) Usage ratio.
\end{minipage}
\Description{The left panel shows total effective stake and active visible identities across reserve prices. The right panel shows the distribution of budget utilization among active records.}
\caption{Illustrative reserve-price counterfactuals and usage ratios under common parameters.}
\label{fig:combined}\label{fig:separable_bidding}\label{fig:usage_ratio}\label{fig:usage}
\end{figure}

Figure~\ref{fig:combined} summarizes the main simulation results.
Panel (a) shows that participation falls as the reserve rises, while total stake varies non-monotonically.
In the illustrative benchmark, lowering the reserve from 32 ETH to 30 ETH increases active identities by 5.922\% and total effective stake by 0.098\%.
Panel (b) reports the corresponding usage ratios, defined as each active address's total equilibrium stake divided by its original budget.
Many addresses retain unused capacity, while the mass at one identifies binding budget constraints.
The reserve-price effect is therefore local: 30 ETH provides a participation-improvement benchmark under these parameters.

\paragraph{Interpretation of the reserve sweep.}

This specification explains the declining identity count at low reserves and the plateau at intermediate reserves in Figure~\ref{fig:combined}(a).
At $\underline r=1$ ETH, the unit reward is approximately $0.105515$, allowing smaller-budget addresses with $c_i=0.105$ to participate.
As the reserve rises, eligibility and profitability constraints reduce the active identity count from 174,725 to 26,762 at 17 ETH, while total stake falls by 3.619\%.
Both quantities and participation remain unchanged over 17--24 ETH, where the minimum active stake, approximately 24.920 ETH, exceeds the reserve.
From 25 ETH, the lower bound starts binding for 22,749 non-Beacon participants.

\FloatBarrier
\section{Discussion}\label{sec:discussion}

The model and simulation give two design lessons. First, the 32 ETH threshold has a local participation cost: in the calibrated data, lowering the threshold to 30 ETH admits marginal addresses and increases active identity count while leaving total effective stake close to the current-threshold benchmark. Second, reserve and cap policies should be evaluated jointly.
The reserve affects both eligibility and the profitability of the minimum stake commitment (see Appendix~\ref{appendix:symmetric_details} for an analysis of a reserve-binding branch).
The cap, in turn, can alter splitting incentives.
Together, these effects determine the security--participation outcomes that the designer evaluates under its preferences and the specified equilibrium-selection rule.

The complete-information benchmark abstracts from uncertainty about ownership and opponents' costs.
In practice, mixers and other privacy tools can obscure transaction links, making common ownership harder to infer~\citep{meiklejohn2018mobius}.
An incomplete-information extension would model owners' beliefs about common control and private costs while retaining joint optimization across each owner's identities. Extensions with MEV, delegation, latency, coalition-level ownership constraints, and alternative issuance schedules would test how robust this frontier is once protocol rewards interact with execution-layer incentives.

\clearpage
\bibliographystyle{RM}
\bibliography{references}

@article{vickrey1961counterspeculation,
  title={Counterspeculation, auctions, and competitive sealed tenders},
  author={Vickrey, William},
  journal={The Journal of finance},
  volume={16},
  number={1},
  pages={8--37},
  year={1961},
  publisher={JSTOR}
}

@article{myerson1981optimal,
  title={Optimal auction design},
  author={Myerson, Roger B},
  journal={Mathematics of operations research},
  volume={6},
  number={1},
  pages={58--73},
  year={1981},
  publisher={INFORMS}
}

@article{maheswaran2003nash,
  title={Nash equilibrium and decentralized negotiation in auctioning divisible resources},
  author={Maheswaran, Rajiv T and Ba{\c{s}}ar, Tamer},
  journal={Group Decision and Negotiation},
  volume={12},
  number={5},
  pages={361--395},
  year={2003},
  publisher={Springer}
}

@inproceedings{zhang2024max,
  title={Max Attestation Matters: Making Honest Parties Lose Their Incentives in Ethereum $\{$PoS$\}$},
  author={Zhang, Mingfei and Li, Rujia and Duan, Sisi},
  booktitle={33rd USENIX Security Symposium (USENIX Security 24)},
  pages={6255--6272},
  year={2024}
}

@inproceedings{chen2024game,
  title={A Game Theoretic Analysis of Validator Strategies in Ethereum 2.0},
  author={Chen, Chien-Chih and Golab, Wojciech},
  booktitle={Proceedings of the 6th ACM International Symposium on Blockchain and Secure Critical Infrastructure},
  pages={1--14},
  year={2024}
}

@article{kelly1997charging,
  title={Charging and rate control for elastic traffic},
  author={Kelly, Frank},
  journal={European transactions on Telecommunications},
  volume={8},
  number={1},
  pages={33--37},
  year={1997},
  publisher={Wiley Online Library}
}

@article{feldman2008proportional,
  title={The proportional-share allocation market for computational resources},
  author={Feldman, Michal and Lai, Kevin and Zhang, Li},
  journal={IEEE Transactions on Parallel and Distributed Systems},
  volume={20},
  number={8},
  pages={1075--1088},
  year={2008},
  publisher={IEEE}
}

@inproceedings{christodoulou2015efficiency,
  title={On the efficiency of the proportional allocation mechanism for divisible resources},
  author={Christodoulou, George and Sgouritsa, Alkmini and Tang, Bo},
  booktitle={International Symposium on Algorithmic Game Theory},
  pages={165--177},
  year={2015},
  organization={Springer}
}

@article{feldman1974fairness,
  title={Fairness and envy},
  author={Feldman, Allan and Kirman, Alan},
  journal={The American Economic Review},
  volume={64},
  number={6},
  pages={995--1005},
  year={1974},
  publisher={JSTOR}
}

@article{roughgarden2021transaction,
  title={Transaction fee mechanism design},
  author={Roughgarden, Tim},
  journal={ACM SIGecom Exchanges},
  volume={19},
  number={1},
  pages={52--55},
  year={2021},
  publisher={ACM New York, NY, USA}
}

@inproceedings{daian2020flash,
  title={Flash boys 2.0: Frontrunning in decentralized exchanges, miner extractable value, and consensus instability},
  author={Daian, Philip and Goldfeder, Steven and Kell, Tyler and Li, Yunqi and Zhao, Xueyuan and Bentov, Iddo and Breidenbach, Lorenz and Juels, Ari},
  booktitle={2020 IEEE symposium on security and privacy (SP)},
  pages={910--927},
  year={2020},
  organization={IEEE}
}

@inproceedings{bachrach2013incentives,
  title={Incentives and efficiency in uncertain collaborative environments},
  author={Bachrach, Yoram and Syrgkanis, Vasilis and Vojnovi{\'c}, Milan},
  booktitle={International Conference on Web and Internet Economics},
  pages={26--39},
  year={2013},
  organization={Springer}
}

@inproceedings{birmpas2022parallel,
  title={Parallel contests for crowdsourcing reviews: existence and quality of equilibria},
  author={Birmpas, Georgios and Kovalchuk, Lyudmila and Lazos, Philip and Oliynykov, Roman},
  booktitle={Proceedings of the 4th ACM Conference on Advances in Financial Technologies},
  pages={268--280},
  year={2022}
}

@article{chen2023gacha,
  title={Gacha game analysis and design},
  author={Chen, Canhui and Fang, Zhixuan},
  journal={Proceedings of the ACM on Measurement and Analysis of Computing Systems},
  volume={7},
  number={1},
  pages={1--45},
  year={2023},
  publisher={ACM New York, NY, USA}
}

@article{saleh2021blockchain,
  title={Blockchain without waste: Proof-of-stake},
  author={Saleh, Fahad},
  journal={The Review of financial studies},
  volume={34},
  number={3},
  pages={1156--1190},
  year={2021},
  publisher={Oxford University Press}
}

@article{alsabah2020pitfalls,
  title={Pitfalls of Bitcoin's Proof-of-Work: R\&D arms race and mining centralization},
  author={Alsabah, Humoud and Capponi, Agostino},
  journal={Available at SSRN 3273982},
  year={2020}
}

@misc{ethsolo2024,
  author       = {{Eth-educators}},
  title        = {Ethereum Solo Staker Dataset},
  year         = {2024},
  note         = {\url{https://github.com/eth-educators/solo-stakers}}
}

@misc{stakecat_solo_2024,
  author       = {{Stake-Cat}},
  title        = {Ethereum Addresses Belonging to Solo Stakers},
  year         = {2024},
  note         = {\url{https://github.com/Stake-Cat/Solo-Stakers}}
}

@inproceedings{brunjes2020reward,
  title={Reward sharing schemes for stake pools},
  author={Br{\"u}njes, Lars and Kiayias, Aggelos and Koutsoupias, Elias and Stouka, Aikaterini-Panagiota},
  booktitle={2020 IEEE european symposium on security and privacy (EuroS\&p)},
  pages={256--275},
  year={2020},
  organization={IEEE}
}

@book{buterin2022proof,
  title={Proof of stake: The making of Ethereum and the philosophy of blockchains},
  author={Buterin, Vitalik},
  year={2022},
  publisher={Seven Stories Press}
}

@misc{ethereumdocs,
  title = {Ethereum Proof-of-Stake Economics},
  author = {{Ethereum Foundation}},
  year = {2023},
  note         = {\url{https://ethereum.org/en/developers/docs/consensus-mechanisms/pos/}}
}

@article{john2025proof,
  title={Proof-of-work versus proof-of-stake: A comparative economic analysis},
  author={John, Kose and Rivera, Thomas J and Saleh, Fahad},
  journal={The Review of Financial Studies},
  volume={38},
  number={7},
  pages={1955--2004},
  year={2025},
  publisher={Oxford University Press}
}

@article{buterin2020combining,
  title={Combining GHOST and casper},
  author={Buterin, Vitalik and Hernandez, Diego and Kamphefner, Thor and Pham, Khiem and Qiao, Zhi and Ryan, Danny and Sin, Juhyeok and Wang, Ying and Zhang, Yan X},
  journal={arXiv preprint arXiv:2003.03052},
  year={2020}
}

@misc{cong2025tokenomics,
  author       = {Cong, Lin William and He, Zhiheng and Tang, Ke},
  title        = {The Tokenomics of Staking},
  howpublished = {NBER Working Paper No. 33640},
  year         = {2025},
  note         = {\url{https://doi.org/10.3386/w33640}}
}

@article{CAN2022274,
title = {On reward sharing in blockchain mining pools},
journal = {Games and Economic Behavior},
volume = {136},
pages = {274-298},
year = {2022},
issn = {0899-8256},
author = {Burak Can and Jens {Leth Hougaard} and Mohsen Pourpouneh}
}

@inproceedings{fanti2019compounding,
  title={Compounding of wealth in proof-of-stake cryptocurrencies},
  author={Fanti, Giulia and Kogan, Leonid and Oh, Sewoong and Ruan, Kathleen and Viswanath, Pramod and Wang, Gerui},
  booktitle={International conference on financial cryptography and data security},
  pages={42--61},
  year={2019},
  organization={Springer}
}

@misc{eip7251,
  author = {Neuder, Mike and Francesc, Francesco and Dapplion and Kalinin, Mikhail and Aditya, Aditya and Drake, Justin and Lightclient and Lange, Felix},
  title = {EIP-7251: Increase the MAX\_EFFECTIVE\_BALANCE},
  year = {2023},
  note         = {\url{https://eips.ethereum.org/EIPS/eip-7251}}
}

@article{johari2004efficiency,
  title={Efficiency loss in a network resource allocation game},
  author={Johari, Ramesh and Tsitsiklis, John N},
  journal={Mathematics of Operations Research},
  volume={29},
  number={3},
  pages={407--435},
  year={2004},
  publisher={INFORMS}
}

@article{johari2009efficiency,
  title={Efficiency of scalar-parameterized mechanisms},
  author={Johari, Ramesh and Tsitsiklis, John N},
  journal={Operations Research},
  volume={57},
  number={4},
  pages={823--839},
  year={2009},
  publisher={INFORMS}
}

@inproceedings{amanatidis2023partial,
  title={Partial Allocations in Budget-Feasible Mechanism Design: Bridging Multiple Levels of Service and Divisible Agents},
  author={Amanatidis, Georgios and Klumper, Sophie and Markakis, Evangelos and Sch{\"a}fer, Guido and Tsikiridis, Artem},
  booktitle={International Conference on Web and Internet Economics},
  pages={41--58},
  year={2023},
  organization={Springer}
}

@inproceedings{seddighin2019externalities,
  title={Externalities and fairness},
  author={Seddighin, Masoud and Saleh, Hamed and Ghodsi, Mohammad},
  booktitle={The World Wide Web Conference},
  pages={538--548},
  year={2019}
}

@article{garimidi2025tullock,
  title={Tullock contests in the wild: Applications in blockchains},
  author={Garimidi, Pranav and Neuder, Michael and Roughgarden, Tim},
  journal={ACM SIGecom Exchanges},
  volume={23},
  number={1},
  pages={24--34},
  year={2025}
}

@article{meiklejohn2018mobius,
  title={M{\"o}bius: Trustless tumbling for transaction privacy},
  author={Meiklejohn, Sarah and Mercer, Rebekah},
  year={2018},
  publisher={UCL Department of Computer Science}
}

@misc{beccuti2025stakingmarket,
  author       = {Beccuti, Juan and Chantramonklasri, Thunj
                  and Hafner, Matthias and Oderbolz, Nicolas},
  title        = {Towards a Formal Framework of the {Ethereum} Staking Market},
  year         = {2025},
  howpublished = {arXiv:2503.14385, version 2},
  url          = {https://arxiv.org/abs/2503.14385v2}
}

@misc{lidoStakingDocs,
  author       = {{Lido}},
  title        = {{Lido} Documentation: Introduction},
  year         = {n.d.},
  note         = {\url{https://docs.lido.fi/}}
}

\clearpage
\appendix
\section{Background}
\subsection{Nothing-at-stake Attack} \label{appendix:nothing-at-stake}
In Proof-of-Stake (PoS) systems, one major concern is the \textit{nothing-at-stake} attack, where a validator can sign blocks on multiple competing forks at little additional computational cost.
To mitigate this issue, Casper requires validators to stake ETH and makes specified conflicting votes slashable, including double votes for the same target epoch and surround votes~\citep{ethereumdocs}.
Routine missed duties instead incur missed rewards or ordinary penalties.
This mechanism highlights the role of staking in securing PoS consensus.
\citet{saleh2021blockchain} uses an economic model to argue that fork disagreement reduces the value of the native coin held by validators, and that PoS can reach consensus if minimum stakes are sufficiently large or rewards sufficiently modest.

\subsection{Tullock Benchmark and Ethereum Effective Stake}\label{appendix:tullock-benchmark}
The closest mathematical benchmarks are allocation-rate mechanisms and Tullock contests. In a canonical allocation-rate formulation, participant \(i\) chooses effort \(e_i\), receives
\[
    x_i(\mathbf e)=\frac{e_i}{\sum_j e_j},
\]
and obtains payoff
\[
    u_i^{\mathrm{AR}}(\mathbf e)=\mathcal V_i(x_i(\mathbf e))-C_i(e_i),
    \qquad \mathcal V_i'>0,\quad \mathcal V_i''\le 0.
\]

The allocation value is therefore increasing and concave in the received share. This is the structure used by allocation-rate and scalar-parameterized mechanism results~\citep{johari2004efficiency,johari2009efficiency}. A Tullock contest uses the same proportional-selection idea. For a positive prize value $v_i>0$, the simplest case has
\[
    p_i(\mathbf e)=\frac{e_i}{\sum_j e_j},\qquad
    u_i^{\mathrm{Tul}}(\mathbf e)=v_i p_i(\mathbf e)-C_i(e_i),
\]
so effort buys a higher winning probability, and the value side is monotone in the allocation or winning-probability share.

Ethereum staking shares this proportional-allocation rule, but the value of a given share also depends on total stake.
In the static setting, each active validator controls one identity and has payoff
\[
 V_i(Q_i;Q_{-i})
 =\alpha\frac{Q_i}{\sqrt{Q_i+Q_{-i}}}-c_iQ_i-\frac{\kappa_i}{2}Q_i^2.
\]

Writing $Q=Q_i+Q_{-i}$ and $x_i=Q_i/Q$, its gross reward is $\alpha x_i\sqrt Q$.
A fixed allocation-value function $\mathcal V_i(x_i)$ alone therefore cannot represent rewards across different aggregate stakes.
We use the auction framework of ~\citet{maheswaran2003nash} to derive best responses and scalar market clearing while retaining this dependence.
The dynamic model further allows one validator to coordinate several identities subject to separate balance constraints.

The marginal payoff in the static game is
\[
 \frac{\partial V_i}{\partial Q_i}
 =\alpha\frac{Q_i/2+Q_{-i}}{(Q_i+Q_{-i})^{3/2}}-c_i-\kappa_iQ_i.
\]

It can be positive at low stake and negative at high stake as rewards compete with the cost of committed funds.
Net payoff can also be nonmonotone in a Tullock contest, so this property alone does not distinguish the models.
For comparison, the idealized PoW boundary case has payoff
\[
 U_i^{PoW}(a_i;a_{-i})
 =R\frac{a_i}{a_i+a_{-i}}-c_i a_i,
\]
where $a_i$ is the miner's hash power, $a_{-i}=\sum_{j\ne i}a_j$, and $R>0$ is a fixed reward budget over a common accounting period.
Here $c_i a_i$ is mining cost over the same period.
Additional hash power raises winning probability but also raises mining cost.
The aggregate prize here remains fixed, in contrast to the stake-dependent reward pool in PoS.
Appendix~\ref{pow section} develops this boundary case.

\section{Economic Foundation for Cost}\label{appendix:cost_foundation}

The equilibrium analysis takes $C_i(Q)=c_iQ+\kappa_iQ^2/2$ as a primitive.
This appendix explains one environment in which the quadratic term arises from a cost of committing liquidity, even for a risk-neutral validator.
It provides an economic interpretation of the quadratic cost specification.

Let $W_i>0$ with $W_i\ge B_i$ be validator $i$'s liquid resources before the staking decision, measured in ETH-equivalent units.
These resources include the funds available for staking, so $W_i$ is not an additional endowment on top of $B_i$.
The case $W_i=B_i>0$ is admissible, while a zero-budget validator cannot stake.
Committing $0\le Q\le B_i$ leaves the immediately available buffer $W_i-Q$.
During an interval in which the committed stake cannot be recovered in time, a random liquidity need $X_i\ge0$ with finite mean arrives.
A unit funding shortfall incurs an incremental financing or unmet-payment loss $h_i\ge0$.
This delayed-access interpretation is consistent with withdrawal queues, whose duration depends on exit demand.\footnote{\url{https://ethereum.org/staking/withdrawals/}.}

Relative to keeping the same resources liquid, the extra expected loss is
\[
 \mathcal C_i(Q)=h_i\left\{
 \mathbb E[(X_i-W_i+Q)_+]-\mathbb E[(X_i-W_i)_+]
 \right\},\qquad (x)_+=\max\{x,0\}.
\]

The subtraction removes any shortfall that would have occurred without staking.
For each realized need, the positive-part function is nondecreasing and convex in $Q$.
Taking expectations preserves both properties, and $\mathcal C_i(0)=0$.
When $X_i$ has a continuous distribution function $F_i$, dominated differentiation gives
\[\mathcal C_i'(Q)=h_i[1-F_i(W_i-Q)].\]

The next unit of stake is more costly because a smaller remaining buffer makes a shortfall more likely.
If a density $g_i$ exists at the relevant point, the curvature is
$\mathcal C_i''(Q)=h_i g_i(W_i-Q)\ge0$.
Convexity therefore does not require quadratic preferences, independent validator failures, or a quadratic protocol penalty.

For the illustrative distribution $X_i\sim U[0,W_i]$, no shortfall occurs at $Q=0$ and
\[
 \mathcal C_i(Q)
 =\frac{h_i}{W_i}\int_{W_i-Q}^{W_i}(x-W_i+Q)\,dx
 =\frac{h_i}{2W_i}Q^2,\qquad
 \kappa_i=\frac{h_i}{W_i}.
\]

Adding a distinct linear opportunity cost $c_iQ$ gives the cost used in the main text.
The linear component should not count the same financing loss a second time.
For other distributions the resulting cost is generally convex but not exactly quadratic.
Where the expected loss is twice continuously differentiable, a quadratic can instead serve as a local approximation over the stake range of interest.

All identities operated by validator $i$ share the same available resources and funding need.
The relevant commitment is consequently $Q_i=\sum_\ell s_{i\ell}$, and its quadratic cost is
\[
 \frac{\kappa_i}{2}\left(\sum_\ell s_{i\ell}\right)^2
 =\frac{\kappa_i}{2}\sum_\ell s_{i\ell}^2
   +\kappa_i\sum_{\ell<j}s_{i\ell}s_{ij}.
\]

Deleting the cross terms would make a shared funding cost fall when the validator merely creates more keys.
No additional liquidity has been created by that action.

The funding model supplies a possible foundation for convexity, while the main theorems take the quadratic cost directly as a primitive.

\section{Proofs in Section 3}\label{appendix:static_proofs}

\subsection{Proof of Proposition~\ref{prop opt}}\label{section proof 1}
\begin{proof}
Without reserve and upper-bound constraints, the payoff is continuous on $[0,\infty)$ and tends to $-\infty$ as $Q_i\to\infty$ because $c_i>0$.
Strict concavity gives a unique maximizer $Q_i^\dagger\ge0$.
If $Q_{-i}>0$, the derivative at zero is $\alpha/\sqrt{Q_{-i}}-c_i$, giving a zero maximizer exactly when this is nonpositive.
For $Q_{-i}=0$, the continuously extended payoff is $\alpha\sqrt{Q_i}-C_i(Q_i)$ with value zero at the origin and infinite positive right derivative there.
In either case, a positive maximizer is the unique solution of Equation~\eqref{first order}.

For the constrained problem, if $B_i=\underline r$, the projection formula gives the only feasible bid.
For $B_i>\underline r$, form the Lagrangian
\[
 \mathcal L(Q_i,\lambda_1,\lambda_2)
 =V_i(Q_i;Q_{-i})+\lambda_1(B_i-Q_i)+\lambda_2(Q_i-\underline r).
\]

The KKT conditions are
\[
\begin{aligned}
 &\frac{\alpha(Q_{-i}+Q_i/2)}{(Q_i+Q_{-i})^{3/2}}-c_i-\kappa_iQ_i-\lambda_1+\lambda_2=0,\\
 &\underline r\le Q_i\le B_i,\qquad \lambda_1,\lambda_2\ge0,\\
 &\lambda_1(B_i-Q_i)=0,\qquad \lambda_2(Q_i-\underline r)=0.
\end{aligned}
\]

Strict concavity makes these conditions sufficient as well as necessary.
If $\underline r\le Q_i^\dagger\le B_i$, the unconstrained maximizer is feasible and both multipliers may be zero.
If $Q_i^\dagger>B_i$, the payoff is increasing on $[\underline r,B_i]$, so the upper bound is optimal with $\lambda_1=\partial V_i(B_i;Q_{-i})/\partial Q_i\ge0$ and $\lambda_2=0$.
If $Q_i^\dagger<\underline r$, the payoff is decreasing on $[\underline r,B_i]$, so the lower bound is optimal with $\lambda_2=-\partial V_i(\underline r;Q_{-i})/\partial Q_i\ge0$ and $\lambda_1=0$.
Together these cases give Equation~\eqref{eq:static_optimal_bid}.
\qedhere
\end{proof}

For each $i\in A$, the active quantity satisfies $Q_i\ge \underline r>0$, and the payoff can be written as
\[
 V_i(Q_i;Q_{-i})
 =Q_i\left(\frac{\alpha}{\sqrt Q}-c_i-\frac{\kappa_i}{2}Q_i\right).
\]

The outside payoff is zero, although exit is not available in this conditional quantity game.
Dividing $V_i\ge0$ by $Q_i>0$ gives Equation~\eqref{eq:static_IR}, proving both necessity and sufficiency.

For fixed opponents' stake, the constrained best response is weakly increasing in $B_i$ because a larger bid limit expands the feasible set.
The dynamic section examines how additional identities change the validator's feasible stake, holding its budget fixed.

\subsection{Proof of Proposition~\ref{prop:price_demand}}\label{section proof 2}
\begin{proof}
For an interior best response, write $Q=Q_i+Q_{-i}$ and $x_i=Q_i/Q$.
Substituting $Q_i=x_iQ$ and $Q_{-i}=(1-x_i)Q$ into Equation~\eqref{first order} gives
\[
 \alpha\frac{(1-x_i)Q+\tfrac12x_iQ}{Q^{3/2}}
 =c_i+\kappa_iQx_i,
\]
which simplifies to
\[
 p(Q)\left(1-\frac{x_i}{2}\right)=c_i+\kappa_iQx_i,
 \qquad p(Q)=\frac{\alpha}{\sqrt Q}.
\]

Collecting the terms in $x_i$ yields
\[
 \left(\frac{p(Q)}{2}+\kappa_iQ\right)x_i=p(Q)-c_i.
\]

Hence the unconstrained share implied by the first-order condition is
\[
 d_i^0(Q)=\frac{p(Q)-c_i}{p(Q)/2+\kappa_iQ}.
\]

The implied unconstrained bid is $Qd_i^0(Q)$.
Applying the reserve and bid-limit constraints and dividing by total stake gives
\[
 d_i(Q) =\frac{1}{Q}\operatorname{proj}_{[\underline r,B_i]}\left(Qd_i^0(Q)\right)=\max\left\{\frac{\underline r}{Q},\min\left\{d_i^0(Q),\frac{B_i}{Q}\right\}\right\}.
\]

This is Equation~\eqref{eq:static_demand}.
To connect the projection to strategic best responses, write the marginal payoff in share coordinates as
\begin{equation}\label{eq:proof_psi}
 \Psi_i(Q,x)=p(Q)(1-x/2)-c_i-\kappa_iQx.
\end{equation}

At a consistent profile, $\Psi_i(Q,x_i)$ equals the marginal payoff from changing $Q_i$ while holding opponents' stake fixed.
For a given candidate total, its slope in the share is
\[
 \frac{\partial\Psi_i(Q,x)}{\partial x}
 =-\frac{p(Q)}{2}-\kappa_iQ<0.
\]

Projection therefore gives a nonpositive marginal payoff at the lower bound, zero at an interior solution, or a nonnegative marginal payoff at the upper bound.
These are precisely the constrained optimality conditions, which are sufficient by strict concavity.
When $\underline r=B_i$, the projection gives the only feasible bid.

Continuity follows from the positive denominator of $d_i^0(Q)$ and continuity of the projection.
For monotonicity, rewrite the unconstrained share as
\[
 d_i^0(Q)=
 \frac{1-(c_i/\alpha)\sqrt Q}{1/2+(\kappa_i/\alpha)Q^{3/2}}.
\]

Where $p(Q)>c_i$, the numerator is positive and strictly decreasing, while the denominator is positive and nondecreasing.
Thus $d_i^0$ decreases on this region.
Where $p(Q)\le c_i$, the raw share is nonpositive and the reserve sets $d_i(Q)=\underline r/Q$.
The boundary branches $\underline r/Q$ and $B_i/Q$ also decrease, so the projected share is nonincreasing throughout $Q>0$.

Finally, if $\sum_{i\in A}d_i(Q)=1$, setting $Q_i=Qd_i(Q)$ gives
\[
 \sum_{i\in A}Q_i=Q\sum_{i\in A}d_i(Q)=Q,
 \qquad Q_{-i}=Q-Q_i.
\]

The induced bids are feasible and satisfy the strategic optimality conditions against these opponents' totals, hence form a quantity equilibrium.
Conversely, at any quantity equilibrium, the same marginal conditions imply $Q_i/Q=d_i(Q)$.
Summing these shares gives the clearing condition, completing the equivalence.
\qedhere
\end{proof}

\subsection{Proof of Theorem~\ref{thm: unique NE}}\label{section proof theorem unique}
\begin{proof}
We prove the result for fixed intervals $[a_i,b_i]$ with $0<a_i\le b_i<\infty$, which include the static constraints when $a_i=\underline r$ and $b_i=B_i$.
Use Equation~\eqref{eq:static_demand} with $a_i$ in place of $\underline r$ and $b_i$ in place of $B_i$, and define
\[
 D_A(Q)=\sum_{i\in A}d_i(Q),\qquad
 a=\sum_{i\in A}a_i,\qquad b=\sum_{i\in A}b_i.
\]

\textbf{Existence of a clearing total.}
Proposition~\ref{prop:price_demand} gives continuity of every demand share and hence of $D_A$.
The lower and upper bid constraints imply
\[
 \frac{a_i}{Q}\le d_i(Q)\le\frac{b_i}{Q},\qquad Q>0.
\]

Evaluating at the two aggregate bounds yields
\[
 D_A(a)\ge\sum_{i\in A}\frac{a_i}{a}=1,
 \qquad
 D_A(b)\le\sum_{i\in A}\frac{b_i}{b}=1.
\]

The intermediate value theorem therefore gives $Q_A\in[a,b]$ with $D_A(Q_A)=1$.
If $a=b$, the only feasible total already clears demand.
For the static constraints $a_i=\underline r$ and $b_i=B_i$, the resulting bound is $|A|\underline r\le Q_A\le\sum_{i\in A}B_i$.

\textbf{Correspondence with Nash equilibrium.}
Given a clearing total, set $Q_i^*=Q_A d_i(Q_A)$.
These bids satisfy
\[
 Q_i^*\in[a_i,b_i],\qquad
 \sum_{i\in A}Q_i^*=Q_A\sum_{i\in A}d_i(Q_A)=Q_A.
\]

Thus each validator faces the opponents' total $Q_{-i}^*=Q_A-Q_i^*$ used in the demand representation.
By Proposition~\ref{prop:price_demand}, its bid satisfies the constrained best-response conditions against those opponents.
The vector $\mathbf Q^*$ is therefore a Nash equilibrium.
Conversely, let $\widetilde{\mathbf Q}$ be any quantity equilibrium and write $\widetilde Q=\sum_{i\in A}\widetilde Q_i$.
The same constrained optimality conditions imply
\[
 \widetilde Q_i=\widetilde Q\,d_i(\widetilde Q),
 \qquad
 D_A(\widetilde Q)=\sum_{i\in A}\frac{\widetilde Q_i}{\widetilde Q}=1.
\]

Hence every Nash equilibrium corresponds to a clearing total.

\textbf{Uniqueness.}
Let $Q^{(1)}$ be a clearing total and consider any $Q^{(2)}>Q^{(1)}$.
For validator $i$, put $x=d_i(Q^{(1)})$.
Clearing and positive lower bounds give $0<x\le1$.
If the upper bound determines this share, then
\[
 x=\frac{b_i}{Q^{(1)}},\qquad
 d_i(Q^{(2)})\le\frac{b_i}{Q^{(2)}}<x.
\]

Otherwise, $x<b_i/Q^{(1)}$ and the lower-bound or interior optimality condition gives $\Psi_i(Q^{(1)},x)\le0$.
For this fixed share, Equation~\eqref{eq:proof_psi} satisfies
\[
 \frac{\partial\Psi_i(Q,x)}{\partial Q}
 =-\frac{\alpha}{2Q^{3/2}}\left(1-\frac{x}{2}\right)-\kappa_i x<0.
\]

It follows that $\Psi_i(Q^{(2)},x)<0$.
Since $\Psi_i(Q^{(2)},\cdot)$ is strictly decreasing, its zero satisfies $d_i^0(Q^{(2)})<x$.
The new lower-bound share also satisfies $a_i/Q^{(2)}<a_i/Q^{(1)}\le x$.
Consequently,
\[
 d_i(Q^{(2)})
 =\max\left\{\frac{a_i}{Q^{(2)}},
 \min\left\{d_i^0(Q^{(2)}),\frac{b_i}{Q^{(2)}}\right\}\right\}<x.
\]

Thus every validator's demanded share is strictly smaller at $Q^{(2)}$ than at the clearing total $Q^{(1)}$.
Summing gives
\[
 D_A(Q^{(2)})<\sum_{i\in A}d_i(Q^{(1)})=1.
\]

No second clearing total can therefore exist.
The unique total determines each bid through $Q_i^*=Q_A d_i(Q_A)$, proving uniqueness of the Nash equilibrium.
\qedhere
\end{proof}

\section{Further Analysis of Identity Restructuring}

\subsection{Quantity Reduction and Reset}\label{appendix:quantity_reset_proof}
\begin{proof}[Lemma~\ref{lem:reduction} and Corollary~\ref{thm:reset}]
Summing the per-identity minimum stakes gives $L_i\underline r\le Q_i$, and the validator budget gives $Q_i\le B_i$.
With a cap, the per-identity upper bounds also require $Q_i\le L_i\bar r$.
Conversely, every total in $[a_i,b_i^{L_i}]$ can be allocated equally, with each identity receiving $Q_i/L_i$ and satisfying the applicable balance bounds.
The payoff depends on the allocation only through $Q_i$.
This proves the reduction and the exact coverage of continuous allocations.

With counts fixed, the reduced payoff has the same form as the static payoff, with interval $[a_i,b_i^{L_i}]$.
The count menu ensures a positive, finite, nonempty interval in both settings, using $\bar r\ge\underline r$ when a cap is imposed.
The arbitrary-interval proof of Theorem~\ref{thm: unique NE} therefore applies.
Solving its marginal equation gives Equation~\eqref{eq:quadratic_reset}, and aggregate clearing and uniqueness follow.
There is no quantity choice at the all-zero profile, whose payoff vector is defined as zero.
\qedhere
\end{proof}

\subsection{Symmetric Quantities and Identity Restructuring}\label{appendix:symmetric_proofs}

\begin{proof}[Proposition~\ref{prop:symmetric}]
Relabeling any two validators with identical parameters and the same count $\ell$ leaves the quantity game unchanged.
By Corollary~\ref{thm:reset}, its unique equilibrium must be invariant to such relabeling, so all validators choose the same total quantity.
At a common quantity $q$, the marginal payoff is $F_{n,\ell}(q)$ from Equation~\eqref{eq:symmetric_marginal}.
This function strictly decreases with $q$, so its endpoint signs or unique interior zero give the quantity stated in the proposition.
Substitution gives aggregate stake, identity count, and validator payoff.

After one validator changes to a feasible $k>0$, relabeling any two of the $n-1$ non-deviating validators still leaves the new game unchanged.
Uniqueness therefore gives a common opponent quantity, yielding Equation~\eqref{eq:asymmetric_reset}.
The same calculation applies to a deviation by any incumbent.
The capped extension in Section~\ref{solution section} follows from the same argument with the corresponding upper bounds.
\qedhere
\end{proof}

Changing the identity count affects a validator through its stake bounds.
At unchanged validator totals, rewards and costs do not change.
A split can therefore leave the quantity equilibrium unchanged: the new bounds may neither force a different stake nor create a profitable quantity deviation.

To see this, suppose the original symmetric equilibrium has an interior quantity $q$ per validator.
If validator $i$ changes from $\ell$ to $k$ identities and
\[
 k\underline r\le q\le b_i^k,
\]
then $q$ remains its best response to the original opponents' stakes.
The other validators also face the same quantities as before, so none needs to adjust.
The original profile remains a quantity equilibrium and, by Corollary~\ref{thm:reset}, is the unique reset.
Total stake and payoffs are unchanged, while $D_{\mathrm{op}}$ changes by $k-\ell$.

When an original bound binds, however, keeping $q$ feasible is not enough.
For example, relaxing a binding cap can make additional stake profitable, so the quantity equilibrium must be recomputed.

\subsection{Calculations for Examples}\label{appendix:examples}

The comparison holds the economic parameters fixed and varies the per-identity cap.
For each cap, we solve the quantity game before and after one validator changes from one identity to two, then compare that validator's continuation payoffs.
We identify when an adopted split improves $(S,D_{\mathrm{op}})$ and finally check whether further restructuring remains profitable and whether the outcome satisfies Assumption~\ref{assumption 1}.

\paragraph{Step 1: Specify the split.}
We consider $n\in\{2,3\}$ economic validators, initially operating one identity each.
Both cases use $\alpha=13/6$, $c=1$, $\kappa=3/20$, $B=2$, and $\underline r=5/8$.
These parameters define the theoretical examples considered below.
They give $C(q)=q+3q^2/40$ and the full count menu $\{0,1,2,3\}$.
Below, $V(q;R)=\alpha q/\sqrt{q+R}-C(q)$ denotes a validator's payoff when opponents stake $R$.
Only one validator changes its count to two, while opponents keep their identity counts and re-optimize their stakes.
The split raises the deviator's minimum total stake from $\underline r=5/8$ to $Y=2\underline r=5/4$.

\paragraph{Step 2: Re-solve stakes before and after the split.}
The additional identity raises the minimum commitment, but all validators still choose their stakes optimally under the new bounds.
Define the marginal payoff
\[
 G(q;R)=\frac{\alpha(R+q/2)}{(q+R)^{3/2}}-1-\frac{3q}{20}.
\]

It strictly decreases in $q$.
Its reward term is maximized over $R\ge0$ at $R=q/2$, so
\[
 G(q;R)\le\frac{\alpha}{(3/2)^{3/2}\sqrt q}-1-\frac{3q}{20}<0
 \qquad\text{for }q\ge\frac54.
\]

Thus a validator with two or three identities stakes its minimum total, $5/4$ or $15/8$, respectively.

Write $t=\bar r$.
For $n\in\{2,3\}$ validators initially operating one identity each, the uncapped quantity $q_n^0$ solves $G(q_n^0;(n-1)q_n^0)=0$.
After one validator splits, write $x_n^*$ for each opponent's common quantity (not a stake share), which solves $G(x_n^*;Y+(n-2)x_n^*)=0$.
Strict concavity and the unique quantity reset give
\[
\begin{aligned}
 &(q_2^0,x_2^*)\simeq(0.9987,0.9548),
 \qquad(q_3^0,x_3^*)\simeq(0.8539,0.7741),\\
 &q_n^{\mathrm{pre}}(t)=\min\{t,q_n^0\},
 \qquad(q_n^{\mathrm{split}}(t),q_n^{\mathrm{opponent}}(t))
 =\left(Y,\min\{t,x_n^*\}\right).
\end{aligned}
\]

These quantities satisfy all lower bounds for the plotted range $0.65\le t\le1.1$.
Both uncapped quantities lie below $Y$, so the split raises the deviator's stake throughout the cap range studied here.
The splitter's lower bound binds, while the opponents are capped exactly when $t<x_n^*$.

\paragraph{Step 3: Compare the splitting validator's payoffs.}
The two quantity equilibria determine the private gain from splitting.
Both the deviator's higher stake and the opponents' quantity responses enter this comparison.
Write $x_n(t)=\min\{t,x_n^*\}$ for each opponent's post-split quantity.
The one-validator split gain is
\begin{equation}\label{eq:cap_benchmark}
 \Delta_n(t)=\frac{\alpha Y}{\sqrt{Y+(n-1)x_n(t)}}-C(Y)
 -\left[\alpha\sqrt{\frac{q_n^{\mathrm{pre}}(t)}{n}}-C(q_n^{\mathrm{pre}}(t))\right].
\end{equation}

Without a cap, the resulting gains are $\Delta_2^0\simeq-0.0008$ and $\Delta_3^0\simeq0.0046$.
Hence the uncapped outcomes after the split decision are $(S,D_{\mathrm{op}})\simeq(1.9974,2)$ for $n=2$ and approximately $(2.7981,4)$ for $n=3$.

For Example~\ref{ex:reset_split_flip}, $n=2$ and $t=0.8$ give pre-split quantities $(0.8,0.8)$ and post-split quantities $(5/4,0.8)$.
The validator gains approximately $0.0021$ from this one-step split, producing $(S,D_{\mathrm{op}})=(41/20,3)$.
Table~\ref{tab:cap_count_audit} collects the payoff comparisons for both examples.

\paragraph{Step 4: Identify the useful cap range.}
Private adoption alone does not ensure that both protocol coordinates improve.
We first identify caps with a positive split gain, then compare the adopted outcomes with the uncapped benchmark.
Over the displayed range $t\in[0.65,1.1]$, the adoption conditions are
\begin{equation}\label{eq:cap_cutoff}
\begin{aligned}
 \Delta_2(t)>0&\quad\Longleftrightarrow\quad t<\tau_2,\\
 \Delta_3(t)>0&\quad\Longleftrightarrow\quad t<\tau_3^-\ \text{or}\ t>\tau_3^+,
\end{aligned}
\end{equation}
where $\tau_2\simeq0.8108$, $\tau_3^-\simeq0.6559$, and $\tau_3^+\simeq0.8438$.

To obtain the cap thresholds, let $f_n(t)$ denote Equation~\eqref{eq:cap_benchmark} on $[5/8,x_n^*]$.
Both pre-split validators and post-split opponents are capped on this interval, yielding
\[
 f_n'(t)=-\frac{\alpha Y(n-1)}{2[Y+(n-1)t]^{3/2}}
 -\frac{\alpha}{2\sqrt{nt}}+1+\frac{3t}{20},
\]
\[
 f_n''(t)=\frac{3\alpha Y(n-1)^2}{4[Y+(n-1)t]^{5/2}}
 +\frac{\alpha}{4\sqrt n\,t^{3/2}}+\frac3{20}>0.
\]

The endpoint derivatives are $f_2'(x_2^*)\simeq-0.0544$ and $f_3'(x_3^*)\simeq-0.1734$.
Thus both gains strictly decrease on their first branches.
They start positive at $5/8$ and become negative at $x_n^*$, giving one zero on each branch.
For $x_n^*\le t\le q_n^0$, the post-split payoff is constant while the symmetric pre-split payoff decreases.
The gain therefore increases to its uncapped value.
That value is negative for $n=2$ and positive for $n=3$, producing the second three-validator zero.
Above $q_n^0$, the gains remain constant.
Root brackets, with endpoints rounded outward to four decimal places, are
\[
\begin{aligned}
 0.8107&<\tau_2<0.8108,\\
 0.6558&<\tau_3^-<0.6559,\\
 0.8437&<\tau_3^+<0.8438.
\end{aligned}
\]

This proves the adoption regions in Equation~\eqref{eq:cap_cutoff}, with no split at the zero-gain boundaries.
For an adopted two-validator split, $t<\tau_2<x_2^*$ gives $S=Y+t$ and $D_{\mathrm{op}}=3$.
Relative to its uncapped outcome $(S_2^0,2)$, where $S_2^0=2q_2^0$, both coordinates strictly improve exactly when $S_2^0-Y<t<\tau_2$.
For $n=3$, the low-cap split region has total stake below the uncapped post-split benchmark and the same identity count.
The intermediate no-split region has three identities, and the high-cap split region reproduces the uncapped coordinates.

\paragraph{Step 5: Check admissibility and remaining deviations.}
A profitable split need not be a terminal equilibrium, and resource feasibility does not guarantee the validator-share condition.
We therefore check both the shares at the resulting profiles and the complete menu of unilateral identity-count choices.

\noindent\textit{Validator shares.}
The reserve remains $5/8$, but the plotted cap starts at $0.65$.
At the former endpoint $t=5/8$, a two-validator split gives totals $(5/4,5/8)$, so the splitter controls $2/3$ of all stake and fails Assumption~\ref{assumption 1}.
At $t=0.65$, the corresponding totals are $(5/4,13/20)$ and the splitter's share is
\[
 \frac{5/4}{5/4+13/20}=\frac{25}{38}<\frac23.
\]

For $0.65\le t\le x_2^*$ this share decreases with $t$, and for larger caps the post-split quantities are constant.
The two-validator pre-split shares are $1/2$.
In the three-validator case, pre-split shares are $1/3$, and all post-split shares are below $1/2$ on the displayed range.
Thus every profile used for the plotted one-split comparisons is admissible.
The cap-$0.8$ example is unchanged: its splitter controls $25/41<2/3$ of total stake.
The former endpoint remains useful only as an algebraic boundary calculation.

\noindent\textit{Remaining count deviations.}
Throughout $S_2^0-Y<t<\tau_2$, the two-validator post-split profile is a pure equilibrium and satisfies Assumption~\ref{assumption 1}.
The splitter prefers two identities to one by $\Delta_2(t)>0$, and its payoff is positive.
A change to three identities raises its quantity to $15/8$ while the opponent stays capped, since the opponent's unconstrained response to $15/8$ is approximately $0.8125>\tau_2$.
Its payoff decreases over that quantity interval because $G(q;R)<0$ for $q\ge5/4$.
The other validator's current payoff increases with $t<x_2^*$ and is at least $V(S_2^0-Y;Y)\simeq0.3565$.
This exceeds its payoff with two identities, $V(5/4;5/4)\simeq0.3457$, while three identities yield less and exit gives zero.

Table~\ref{tab:cap_count_audit} checks every resource-feasible count in the restructuring game and separately flags inadmissible post-deviation profiles.
For $n=3$ without a cap, the splitter prefers two identities, while each of the other validators prefers one.
Thus the resulting profile $(2,1,1)$ is a post-reset equilibrium, as is the capped two-validator profile in Example~\ref{ex:reset_split_flip}.

\begin{table}[htbp]
\centering

\caption{Complete count comparisons for Examples~\ref{ex:uncapped_restructuring} and~\ref{ex:reset_split_flip}.
Each column fixes the other validators' counts at the stated outcome and reports this validator's payoff after choosing $k$ and resetting quantities.
Payoffs are rounded to four decimal places, and bold entries mark the chosen count.
The two nonsplitters in the uncapped $n=3$ outcome have identical payoffs.
A dagger marks a deviation whose reset violates Assumption~\ref{assumption 1}.
Such entries are formal calculations, not payoffs for normal protocol operation.}
\label{tab:cap_count_audit}
\begin{tabular}{crrrrr}
\toprule
$k$ & \shortstack{No cap, $n=2$\\Either validator}
 & \shortstack{No cap, $n=3$\\Splitter}
 & \shortstack{No cap, $n=3$\\Nonsplitter}
 & \shortstack{Cap $0.8$, $n=2$\\Splitter}
 & \shortstack{Cap $0.8$, $n=2$\\Nonsplitter}\\
\midrule
0 & $0.0000^{\dagger}$ & 0.0000 & 0.0000 & $0.0000^{\dagger}$ & $0.0000^{\dagger}$\\
1 & \textbf{0.4576} & 0.2473 & \textbf{0.1836} & 0.5223 & \textbf{0.3626}\\
2 & 0.4568 & \textbf{0.2519} & 0.1601 & \textbf{0.5244} & 0.3457\\
3 & $0.3394^{\dagger}$ & 0.1484 & -0.0408 & $0.3452^{\dagger}$ & 0.1594\\
\bottomrule
\end{tabular}
\end{table}

Figure~\ref{fig:split_incentive_stability} evaluates one unilateral split at each cap, rather than iterating all possible restructurings.
Some very tight-cap two-validator outcomes admit a further profitable split by the other validator and are not terminal post-reset equilibria.
For example, at $t=0.65$, that further split gains approximately $0.0057$, and both its pre- and post-deviation profiles satisfy Assumption~\ref{assumption 1}.
The complete-menu checks retain this distinction, while the beneficial two-validator interval and the representative uncapped outcomes have the equilibrium properties stated above.
Scalar market clearing and direct validator best responses independently verify the quantity resets used in the figure.

\subsection{Proof of Theorem~\ref{thm:post_reset_partition_equilibrium}}\label{section proof post reset partition equilibrium}
\begin{proof}
In the restructuring game of Definition~\ref{def:restructuring-game}, each action set $\mathcal L_i(\underline r)$ is finite and nonempty, and every action satisfies the resource constraints.
Corollary~\ref{thm:reset} assigns a unique quantity and formal payoff vector to each resource-feasible count profile.
The restructuring game is consequently a finite normal-form game.

For completeness, let $\Sigma_i$ be the probability simplex over $\mathcal L_i(\underline r)$.
Expected payoff is continuous on $\prod_i\Sigma_i$ and affine in each validator's own probabilities.
The mixed best-response correspondence has nonempty compact convex values and a closed graph.
Kakutani's fixed-point theorem therefore gives a mixed Nash equilibrium.
Because validators observe the count profile before choosing stakes, every profile defines a quantity subgame.
Combining an equilibrium of the restructuring game with the prescribed quantity equilibrium after every profile establishes subgame perfection.
Equal allocation across a validator's identities can be used as the continuation selection without changing any payoff.

The action $L_i=0$ gives payoff zero for every profile of the opponents' counts.
A pure equilibrium therefore gives every validator a nonnegative payoff, while a mixed equilibrium gives nonnegative expected payoff.
This does not require every realized count profile in a mixed equilibrium to be individually rational.
Since every feasible continuous allocation has a count covered by Lemma~\ref{lem:reduction}, the count comparisons exhaust unilateral structural deviations in this two-stage game.
\qedhere
\end{proof}

\subsection{Designer Objective and Equilibrium Selection}\label{section proof protocol frontier selection}
For every design in $\Theta$, Theorem~\ref{thm:post_reset_partition_equilibrium} gives a nonempty mixed-equilibrium set.
Fix a rule $\eta$ that selects an equilibrium for each design.
Every count profile has finite coordinates and payoff, so the expectation in Equation~\eqref{eq:protocol_problem} is a finite sum with a well-defined value.
There are finitely many designs, so the largest of these values is attained.
This observation specifies the designer's optimization once its preference and the response selection are given.
It does not supply a numerical optimum without those inputs.

If the selected equilibria are pure, let $\mathcal C_\eta$ be their finite set of coordinate pairs.
A pair dominated by another pair in $\mathcal C_\eta$ cannot maximize a strictly increasing $\Phi$.
Hence an optimal design selects a Pareto-efficient pair within $\mathcal C_\eta$.
With mixed responses, the objective is expected $\Phi$ over realized outcomes, generally different from $\Phi$ evaluated at the expected coordinates.
The pure-outcome statement is not asserted for every realization of a mixed equilibrium.

The numerical procedure fixes candidate ordering and initialization and reports a selected pure equilibrium when its complete-menu audit succeeds.
Such a certificate verifies the absence of profitable unilateral count deviations at that outcome.
It does not imply that the equilibrium is unique or that a different selection rule would produce identical protocol coordinates.

\subsection{Reserve-Binding Participation}\label{appendix:symmetric_details}
This subsection studies the branch in which every participating validator commits the reserve through one identity.
More participants reduce the reward per unit, while a higher reserve can raise each participant's average funding cost.
Equilibrium requires incumbents to be willing to stay and entry to be unprofitable.
The characterization below also checks deviations to any feasible identity count.

Suppose all potential validators share $(c,\kappa)$ and satisfy $B_i\ge\underline r$, with possibly different budgets and count menus.
We consider profiles in which each active validator operates one identity and stakes $\underline r$.
Define average cost and the break-even aggregate by
\[
 \overline c(\underline r)=c+\frac{\kappa \underline r}{2},\qquad
 T(\underline r)=\left[\frac{\alpha}{\overline c(\underline r)}\right]^2.
\]

\begin{theorem}[Exact reserve-binding characterization]\label{thm:branch}
Fix $2\le m<M$ and suppose
\begin{equation}\label{eq:marginal_condition}
 p(m\underline r)\left(1-\frac{1}{2m}\right)\le c+\kappa \underline r.
\end{equation}

A profile with $m$ active validators, each operating one identity and staking $\underline r$, is a pure equilibrium against all feasible count deviations if and only if
\[
 m\underline r\le T(\underline r)\le(m+1)\underline r.
\]

For non-integer $T(\underline r)/\underline r$, the participant count in this branch is uniquely
\[
 \lfloor T(\underline r)/\underline r\rfloor.
\]
At integer targets, both adjacent counts can satisfy the interval, subject to the population and marginal conditions.
\end{theorem}

\begin{corollary}[Selected branch outcome]\label{cor:floor}
Let $n=\lfloor T(\underline r)/\underline r\rfloor$ with $2\le n<M$.
The marginal condition holds automatically, and any $n$ validators form a pure post-reset equilibrium with security and participation
\[
 S(\underline r)=n\underline r,\qquad D_{\mathrm{op}}=n,
 \qquad T(\underline r)-\underline r<S(\underline r)\le T(\underline r).
\]
\end{corollary}

For $q>0$ and opponents' total $Q_{-i}\ge0$, define the marginal reward
\[
 \mathcal M(q,Q_{-i})=p(q+Q_{-i})\left(1-\frac{q}{2(q+Q_{-i})}\right).
\]

It strictly decreases with $q$.
For a validator at quantity $\underline r$, define $g(z)=p(z)(1-\underline r/(2z))$.
Then
\[
 g'(z)=-\frac{p(z)}{2z}\left(1-\frac{3\underline r}{2z}\right)<0
 \qquad\text{for }z\ge2\underline r.
\]

These comparisons will verify the quantities after each deviation, rather than presume that opponents hold their stakes fixed.

\begin{proof}[Theorem~\ref{thm:branch}]
At the proposed profile, every active validator has quantity $\underline r$ and total stake is $m\underline r$.
Equation~\eqref{eq:marginal_condition} is its lower-bound optimality condition, so the profile is the unique conditional quantity equilibrium.
Each incumbent's payoff is $\underline r[p(m\underline r)-\overline c(\underline r)]$.

Suppose an incumbent changes its identity count to a feasible $k\ge1$.
Propose that it stakes $k\underline r$ and all other $m-1$ incumbents stake $\underline r$.
Total stake is $z=(m-1+k)\underline r\ge m\underline r\ge2\underline r$.
The other validators' marginal rewards are $g(z)\le g(m\underline r)\le c+\kappa \underline r$, so their lower bounds are best responses.
For the deviator,
\[
 \mathcal M(k\underline r,(m-1)\underline r)\le \mathcal M(\underline r,(m-1)\underline r)=g(m\underline r)
 \le c+\kappa \underline r\le c+\kappa k\underline r.
\]

Its lower bound is also optimal.
Thus these proposed quantities constitute the unique post-deviation reset.
An inactive validator entering with $k\ge1$ identities similarly induces quantities $k\underline r$ for itself and $\underline r$ for the $m$ incumbents.
For this entrant, use
$\mathcal M(k\underline r,m\underline r)\le \mathcal M(\underline r,m\underline r)=g((m+1)\underline r)\le g(m\underline r)$.
The incumbents' marginal conditions follow from the same monotonicity of $g$.

Necessity follows by considering exit and entry with a single identity.
An incumbent must weakly prefer its payoff to zero, which is equivalent to $m\underline r\le T(\underline r)$.
There is at least one inactive validator since $m<M$.
Its payoff from entering with one key is $\underline r[p((m+1)\underline r)-\overline c(\underline r)]$, so no profitable entry requires $T(\underline r)\le(m+1)\underline r$.

For sufficiency, assume $m\underline r\le T(\underline r)\le(m+1)\underline r$.
Incumbents have nonnegative payoffs and cannot gain by exit.
For every $k\ge1$,
\[
 C(k\underline r)=ck\underline r+\frac{\kappa k^2\underline r^2}{2}
       \ge k\left(c\underline r+\frac{\kappa \underline r^2}{2}\right)=kC(\underline r).
\]

An incumbent choosing $k\ge2$ faces $z=(m-1+k)\underline r\ge(m+1)\underline r\ge T(\underline r)$ and earns at most
\[
 k\underline r p(z)-kC(\underline r)=k\underline r[p(z)-\overline c(\underline r)]\le0.
\]

This does not exceed its current nonnegative payoff.
The count $k=1$ leaves the game unchanged.
Every entrant has $z=(m+k)\underline r\ge(m+1)\underline r\ge T(\underline r)$ and also earns at most zero.
All feasible counts have been checked, proving sufficiency.
Finally, the interval $m\underline r\le T(\underline r)\le(m+1)\underline r$ has the stated integer solutions, subject to the maintained marginal condition and population bounds.
\qedhere
\end{proof}

\begin{proof}[Corollary~\ref{cor:floor}]
The definition $n=\lfloor T(\underline r)/\underline r\rfloor$ gives $n\underline r\le T(\underline r)<(n+1)\underline r$.
Moreover,
\[
\begin{aligned}
 p(n\underline r)\left(1-\frac1{2n}\right)
 &=\overline c(\underline r)\sqrt{\frac{T(\underline r)}{n\underline r}}\left(1-\frac1{2n}\right)\\
 &<\overline c(\underline r)\left(1+\frac1{2n}\right)\left(1-\frac1{2n}\right)
 <\overline c(\underline r)\le c+\kappa\underline r.
\end{aligned}
\]

Thus the marginal condition holds, and Theorem~\ref{thm:branch} yields the equilibrium and the stated rounding bound.
\qedhere
\end{proof}

Within this branch, $T(\underline r)$ decreases with the reserve when $\kappa>0$.
Conditional on an unchanged participant count, total stake $S=n\underline r$ increases with the reserve.
Changes in the equilibrium count can instead produce downward jumps, so total stake need not be monotone.
This characterization concerns reserve-binding profiles with one identity per participant and does not identify the participating validators or rule out other identity structures.

\section{Algorithms}
\subsection{Static Algorithm}\label{appendix:static algorithm}
For a fixed nonempty active set $A$, Algorithm~\ref{algorithm 1} iteratively adjusts validator stakes within the bounds $[a_i,b_i]$.
Following~\citet{maheswaran2003nash}, validator $i$ uses its current allocation share $x_{i,t}=Q_{i,t}/\sum_{j\in A}Q_{j,t}$ to form a target stake.
Write $\mathcal Q_i(x)=y_i(x)^2$, where $y_i(x)>0$ solves
\[
 \kappa_i x y_i(x)^3+c_i y_i(x)=\alpha(1-x/2).
\]

The next bid moves toward $x_{i,t}\mathcal Q_i(x_{i,t})$ and is projected onto the feasible interval:
\begin{equation}\label{update bid}
 Q_{i,t+1}=\operatorname{proj}_{[a_i,b_i]}
 \left(\lambda x_{i,t}\mathcal Q_i(x_{i,t})+(1-\lambda)Q_{i,t}\right),
 \qquad 0<\lambda\le1.
\end{equation}

Here $a_i=\underline r$ and $b_i=B_i$ in the static one-identity game.
The same update accommodates the validator-specific bounds following an identity change.

\begin{algorithm}[H]
\caption{Decentralized bidding algorithm}\label{algorithm 1}
\begin{algorithmic}[1]
\Procedure{DecenBid}{$A,\{a_i,b_i\},\alpha,\{c_i,\kappa_i\},\lambda,\epsilon$}
 \State Initialize $Q_{i,0}\gets a_i$ for $i\in A$ and zero otherwise
 \State $t\gets0$, $\text{max\_diff}\gets\infty$
 \While{$\text{max\_diff}>\epsilon$}
  \State $Q_t\gets\sum_{j\in A}Q_{j,t}$
  \For{each $i\in A$}
   \State $x_{i,t}\gets Q_{i,t}/Q_t$
   \State Update $Q_{i,t+1}$ using Equation~\eqref{update bid}
  \EndFor
  \State $\text{max\_diff}\gets\max_{i\in A}|Q_{i,t+1}-Q_{i,t}|$
  \State $t\gets t+1$
 \EndWhile
 \State \Return final stakes and payoffs from Equation~\eqref{utility1}
\EndProcedure
\end{algorithmic}
\end{algorithm}

Participation remains fixed during this iteration, with entry and exit handled by the dynamic algorithm.
An exact fixed point satisfies the constrained optimality conditions, but uniqueness of the quantity equilibrium does not imply convergence of the iteration.
The reported simulations solve the same equilibrium through the scalar clearing equation using Brent's method.

\subsection{Dynamic Algorithm}\label{appendix:dynamic algorithm}
The dynamic algorithm accepts identity-count changes that improve the deviating validator's payoff after stakes are re-optimized.
Let $\mathcal P=(\alpha,\underline r,\bar r,\{B_i,c_i,\kappa_i\}_i)$ collect the model parameters.
For counts $\bm L$, \textsc{Reset} returns the quantity equilibrium and payoffs under
\[
 L_i\underline r\le Q_i\le\min\{B_i,L_i\bar r\}.
\]

The all-inactive profile returns zero quantities and payoffs.
We first consider entry with one identity, exit, and one-identity increases (the routine candidates).
Only if none improves payoff by more than $\epsilon$ do we consider one-identity reductions (the merge candidates).
All candidates must belong to the feasible count menus in Section~\ref{subsec:reset}.

\begin{algorithm}[H]
\caption{Post-reset identity adjustment}\label{algorithm 2}
\begin{algorithmic}[1]
\Procedure{ResetDynamics}{$\bm L^0,\mathcal P,\epsilon$}
 \State $\bm L\gets\bm L^0$, $\text{improved}\gets\text{true}$
 \While{$\text{improved}$}
  \State $(\mathbf Q,\mathbf V)\gets\Call{Reset}{\bm L,\mathcal P}$
  \State $\text{best\_gain}\gets0$, $\text{best\_counts}\gets\bm L$
  \For{phase in $\{\text{routine},\text{merge}\}$, in this order}
   \For{each feasible count change $(i,k)$ in the phase}
    \State $\bm L'\gets(k,\bm L_{-i})$
    \State $(\mathbf Q',\mathbf V')\gets\Call{Reset}{\bm L',\mathcal P}$
    \If{$V_i'-V_i>\text{best\_gain}$}
     \State $\text{best\_gain}\gets V_i'-V_i$, $\text{best\_counts}\gets\bm L'$
    \EndIf
   \EndFor
   \If{$\text{best\_gain}>\epsilon$}
    \State \textbf{break} the phase loop
   \EndIf
  \EndFor
  \State $\text{improved}\gets(\text{best\_gain}>\epsilon)$
  \If{$\text{improved}$}
   \State $\bm L\gets\text{best\_counts}$
  \EndIf
 \EndWhile
 \State \Return $\bm L,\mathbf Q,\mathbf V$, $S=\sum_iQ_i$, and $D_{\mathrm{op}}=\sum_iL_i$
\EndProcedure
\end{algorithmic}
\end{algorithm}

Every candidate holds opponents' identity counts fixed while allowing all validators to adjust their stakes.
The displayed procedure stops when neither candidate group offers a profitable change.
An endpoint is an equilibrium of the full restructuring game only if every feasible count deviation is unprofitable, which is checked separately in the reported computations along with Assumption~\ref{assumption 1}.

\section{Comparison with PoW}\label{pow section}
In the main paper, all of our analysis has focused on PoS. However, the model can be extended to the PoW setting, where miners compute hash values by brute-force search over nonces. In this extension, we assume an idealized environment without adversarial behavior: each miner is honest, meaning that once a valid hash is found, it is immediately broadcast to the network. For simplicity, we further assume that this broadcast incurs no delay.

Under these assumptions, we treat the block race as ending upon discovery, abstracting from propagation and confirmation delays (e.g., the typical six-block confirmation used in practice). With this simplified and honest setting, we can generalize the preceding two-stage analysis to the PoW context.

For a fixed accounting period, assume an exogenous reward budget $R_{\mathrm h}>0$ and a linear cost coefficient $\gamma_i>0$ for miner $i$ over the same period.
Dividing its payoff by $\gamma_i$ gives the cost-normalized counterpart of Equation~\ref{utility1}:
\begin{equation}
    U_i(\mathbf{s})= v(\mathbf{s})-s_i = \alpha_i \cdot \frac{s_i}{s_i+\sum_{j\neq i} s_j}-s_i,
    \label{utilityPoW}  
\end{equation}
where $s_i$ denotes hashing effort and $\alpha_i=R_{\mathrm h}/\gamma_i$ is the reward-to-cost ratio.
This positive scaling preserves each miner's best responses.

This reward allocation in the PoW setting aligns with the approach described in \citet{alsabah2020pitfalls}, where rewards comprising newly minted bitcoins and transaction fees are distributed among miners in proportion to their share of the total network hash power.

The differences between the utility formulation in the PoS setting (Equation~\ref{utility1}) and the PoW setting (Equation~\ref{utilityPoW}) can be summarized as follows:
\begin{itemize}
    \item In PoS, the total reward depends on a term involving the square root of the total balance (i.e., \( \sqrt{\text{total balance}} \)), which does not appear in PoW.
    \item Costs may differ across agents in both models. In the normalized PoW payoff, \( \alpha_i \) captures the reward-to-cost ratio.
\end{itemize}

For a fixed active set of $m$ miners with compact positive effort intervals, the PoW payoffs are continuous and concave in each miner's effort, so a Nash equilibrium exists.
The feasible-set argument in Lemma~\ref{Equivalence of bidding game} shows that the two bidding games
\(
    \left( \{b_i\}_{i=1}^m, \{s_i^*\}_{i=1}^m, m, \mathcal{D} \right)
    \) and \(
    \left( \{s_i^*\}_{i=1}^m, \{s_i^*\}_{i=1}^m, m,\mathcal{D}\right)
    \)
share the equilibrium vector $\mathbf s^*$: restricting upper bounds removes possible deviations without removing that equilibrium.
Here $b_i$ is miner $i$'s effort limit and $\mathcal D$ specifies the reward and effort constraints.
In a block-race interpretation, the success probability $p\in(0,1)$ per attempt can represent mining difficulty, as in \citet{chen2023gacha}.
A larger $p$ means lower difficulty, holding all other factors constant.
Relating $p$ to the fixed-period reward budget requires a model of block duration.
    
Participation throughout a block race, without delayed entry or early exit, requires a separate timing comparison.
We use Markov chain analysis and the framework of gacha games~\citep{chen2023gacha} to examine this question for miners with equal, fixed hashing capacities.

\begin{figure}[h]
\centering
\resizebox{0.9\linewidth}{!}{
\begin{tikzpicture}[
    state/.style={circle, draw=none, fill=blue!20, minimum size=2em, inner sep=2pt},
    arrow/.style={-Stealth, thick},
    redarrow/.style={arrow, red},
    greenarrow/.style={arrow, green!60!black},
    every node/.style={font=\small},
    node distance=2cm and 2.8cm
]

\node[state] (s1) {$S_1$};
\node[state, right=of s1] (s2) {$S_2$};
\node[state, right=of s2] (dots) {$\cdots$};
\node[state, right=of dots] (st) {$S_t$};
\node[state, right=of st] (st1) {$S_{t+1}$};
\node[state, above=1.8cm of dots] (win) {win};
\node[state, below=1.8cm of dots] (quit) {quit};

\draw[arrow] (s1) -- node[midway, above, sloped] {$(1 - p_1),\ -c$} (s2);
\draw[arrow] (s2) -- node[midway, above, sloped] {$(1 - p_2),\ -c$} (dots);
\draw[arrow] (dots) -- node[midway, above, sloped] {$(1 - p_{t-1}),\ -c$} (st);
\draw[arrow] (st) -- node[midway, above, sloped] {$(1 - p_t),\ -c$} (st1);
\draw[arrow] (st1) -- ++(2,0) node[midway, above, sloped] {$(1 - p_{t+1}),\ -c$};

\draw[redarrow] (s1) -- node[midway, above, sloped] {$p_1,\ R - c$} (win);
\draw[redarrow] (s2) -- node[midway, above, sloped] {$p_2,\ R - c$} (win);
\draw[redarrow] (dots) -- node[midway, above, sloped] {$p_t,\ R - c$} (win);
\draw[redarrow] (st) -- node[midway, above, sloped] {} (win);
\draw[redarrow] (st1) -- node[midway, above, sloped] {} (win);

\draw[greenarrow] (s1) -- node[midway, below, sloped] {100\%,\ 0} (quit);
\draw[greenarrow] (s2) -- node[midway, below, sloped] {} (quit);
\draw[greenarrow] (dots) -- node[midway, below, sloped] {} (quit);
\draw[greenarrow] (st) -- node[midway, below, sloped] {} (quit);
\draw[greenarrow] (st1) -- node[midway, below, sloped] {} (quit);

\end{tikzpicture}
}
\Description{A Markov decision process shows a sequence of states with choices to continue pulling or quit, and transitions to an end state with the corresponding rewards and costs.}
\caption{Markov Decision Process. Adapted from \citet{chen2023gacha}.}
\label{gacha game example}
\end{figure}
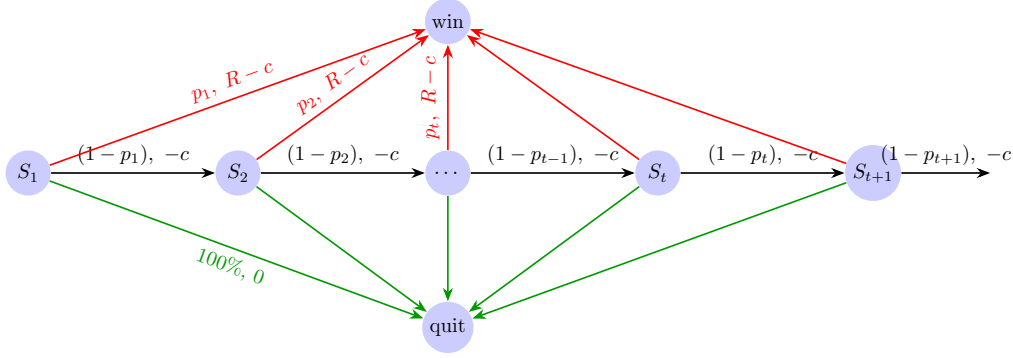

Figure~\ref{gacha game example} illustrates the Markov decision process (MDP) structure.
Here $S_t$ is the state after $t-1$ unsuccessful pulls, before pull $t$, with trial-specific success probability $\rho_t$.
The prize is $R>0$ and each attempted pull costs $c\ge0$.
At each state, the player faces a decision:
\begin{itemize}
    \item \textbf{Pull:} With probability $p_t$, the player wins and moves to a terminal winning state, receiving a net reward of $R-c$. Otherwise, with probability $1-p_t$, the player pays a cost $c$ and transitions to the next pull state $S_{t+1}$.
    \item \textbf{Quit:} The player can choose to quit at any time, receiving a payoff of zero.
\end{itemize}

Then to make the question simple, the authors in \citet{chen2023gacha} consider an agent to make a decision at state $S_1$ and with the whale property, the agent have will either participate the game until win or choose not to in the game at the beginning. Further, if in all the stage the winning probability for each pull is same, we denote $\tilde{p}$ here ,we can have a simplified version of MDP as Figure \ref{simplified} shows:

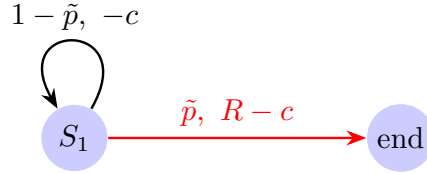
\begin{figure}[h]
\centering
\resizebox{0.4\linewidth}{!}{
\begin{tikzpicture}[
    state/.style={circle, draw=none, fill=blue!20, minimum size=2em, inner sep=2pt},
    arrow/.style={-Stealth, thick},
    redarrow/.style={arrow, red},
    selfloop/.style={arrow, loop above, min distance=12mm, in=120, out=60},
    every node/.style={font=\small},
    node distance=2.5cm and 3cm
]

\node[state] (s1) {$S_1$};
\node[state, right=of s1] (end) {end};

\draw[selfloop] (s1) to node[above] {$1-\tilde{p},\ -c$} (s1);

\draw[redarrow] (s1) -- node[midway, above, sloped] {$\tilde{p},\ R-c$} (end);
\end{tikzpicture}
}
\Description{One active state has a self-loop with failure probability and a cost, and a transition to the end state with success probability and net reward.}
\caption{Simplified version of MDP.}
\label{simplified}
\end{figure}

We now apply this framework to the PoW setting, analyzing the mining process from the perspective of a single block.
Assume there are \(m\ge2\) miners with identical, fixed hashing capacity who repeatedly attempt to find a block.
Attempts are independent, with symmetric tie-breaking if several miners succeed.
Each miner has a fixed success probability \(p\in(0,1)\) per round.
The probability that the block is mined in any given round is $\tilde p=1-(1-p)^m$.
With the block reward normalized to 1, each miner's expected reward conditional on a successful round is $1/m$.
Let $c_i\ge0$ now denote miner $i$'s cost per attempted round, held constant across all participation strategies in this block-race model.
Its unconditional expected payoff per active round is $\tilde p/m-c_i$.
Figure~\ref{simplified} represents this payoff calculation when its terminal reward $R$ is interpreted as the conditional expected reward $1/m$ and its cost $c$ as $c_i$.

We compare participation strategies chosen at the initial state $S_1$, corresponding to the beginning of a block.
The other $m-1$ miners continue with fixed capacity until block discovery, and there is no discounting or new information apart from whether a block has been found.
The independent trials make the remaining race identical after each unsuccessful round.

In this setting, miner $i$ may delay entry for $k\ge1$ rounds or stop after at most $k$ attempted rounds, where $k$ is an integer.
Delayed entry occurs at $S_{k+1}$ only if the block has not already been found.
In both diagrams, $S_t$ denotes the state before round $t$, conditional on the race continuing.
We write $p_1=1-(1-p)^{m-1}$ for block discovery while miner $i$ is absent, and $p_2=1-(1-p)^m$ while all $m$ miners participate.
The following two diagrams represent these choices.

\begin{figure}[htbp]
\centering
\begin{minipage}[b]{0.48\textwidth}
\centering
\resizebox{\linewidth}{!}{
\begin{tikzpicture}[
    state/.style={circle, draw=none, fill=blue!20, minimum size=2em, inner sep=2pt},
    arrow/.style={-Stealth, thick},
    stageone/.style={arrow, red, dotted},
    stagetwo/.style={arrow, blue},
    selfloop/.style={arrow, loop above, min distance=12mm, in=120, out=60},
    every node/.style={font=\small},
    node distance=2cm and 2.8cm
]

\node[state] (s1) {$S_1$};
\node[state, right=of s1] (s2) {$S_2$};
\node[state, right=of s2] (dots) {$\cdots$};
\node[state, right=of dots] (skm1) {$S_{k}$};
\node[state, right=of skm1] (sk) {$S_{k+1}$};
\node[state, below=2.2cm of dots] (end) {end};

\draw[stageone] (s1) -- node[midway, above, sloped] {$1 - p_1,\ 0$} (s2);
\draw[stageone] (s2) -- node[midway, above, sloped] {$1 - p_1,\ 0$} (dots);
\draw[stageone] (dots) -- node[midway, above, sloped] {$1 - p_1,\ 0$} (skm1);
\draw[stageone] (skm1) -- node[midway, above, sloped] {$1 - p_1,\ 0$} (sk);

\draw[stageone] (s1) -- node[midway, below, sloped] {$p_1,\ 0$} (end);
\draw[stageone] (s2) -- node[midway, below, sloped] {$p_1,\ 0$} (end);
\draw[stageone] (dots) -- node[midway, below, sloped] {$p_1,\ 0$} (end);
\draw[stageone] (skm1) -- node[midway, below, sloped] {$p_1,\ 0$} (end);

\draw[selfloop, stagetwo] (sk) to node[above] {$1-p_2,\ -c_i$} (sk);
\draw[stagetwo] (sk) -- node[midway, below, sloped] {$p_2,\ \frac{1}{m}-c_i$} (end);

\end{tikzpicture}
}
\par\centering (a) Delay mining
\label{fig:delay-mining-mdp}
\end{minipage}\hfill
\begin{minipage}[b]{0.48\textwidth}
\centering
\resizebox{\linewidth}{!}{
\begin{tikzpicture}[
    state/.style={circle, draw=none, fill=blue!20, minimum size=2em, inner sep=2pt},
    bluearrow/.style={-Stealth, thick, blue},
    every node/.style={font=\small},
    node distance=2.5cm and 2.8cm
]

\node[state] (s1) {$S_1$};
\node[state, right=of s1] (s2) {$S_2$};
\node[state, right=of s2] (dots) {$\cdots$};
\node[state, right=of dots] (skm1) {$S_{k-1}$};
\node[state, right=of skm1] (sk) {$S_{k}$};
\node[state, below=2.5cm of dots] (end) {end};

\draw[bluearrow] (s1) -- node[midway, above, sloped] {$1-p_2,\ -c_i$} (s2);
\draw[bluearrow] (s2) -- node[midway, above, sloped] {$1-p_2,\ -c_i$} (dots);
\draw[bluearrow] (dots) -- node[midway, above, sloped] {$1-p_2,\ -c_i$} (skm1);
\draw[bluearrow] (skm1) -- node[midway, above, sloped] {$1-p_2,\ -c_i$} (sk);

\draw[bluearrow] (s1) -- node[midway, below, sloped] {$p_2,\ \frac{1}{m}-c_i$} (end);
\draw[bluearrow] (s2) -- node[midway, below, sloped] {$p_2,\ \frac{1}{m}-c_i$} (end);
\draw[bluearrow] (dots) -- node[midway, below, sloped] {$p_2,\ \frac{1}{m}-c_i$} (end);
\draw[bluearrow] (skm1) -- node[midway, below, sloped] {$p_2,\ \frac{1}{m}-c_i$} (end);

\draw[bluearrow] (sk) -- node[midway, below, sloped] {$1,\ p_2\cdot\frac{1}{m}-c_i$} (end);

\end{tikzpicture}
}
\par\centering (b) Early termination
\label{fig:early-terminate-mdp}
\end{minipage}
\Description{Two panels compare delay mining and early termination. The delay strategy stays inactive for an initial number of rounds before joining the mining race. The early-termination strategy stops after a fixed number of attempts.}
\caption{Agent's behavior in the MDP.}
\label{fig:pow-timing-mdp}
\end{figure}
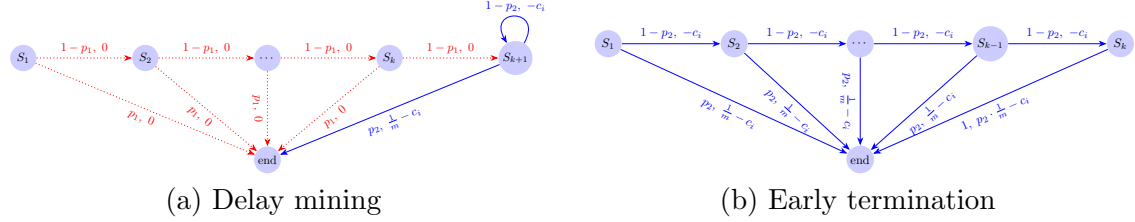

Figure~\ref{fig:pow-timing-mdp}(a) illustrates the strategic decision at state $S_1$, where miner $i$ considers delaying mining until state $S_{k+1}$. The red dotted line represents the \textit{delay strategy}: here, only the other $m-1$ miners participate, resulting in zero cost for miner $i$. The block may be found by the other miners before the delayed entry, so reaching $S_{k+1}$ requires $k$ unsuccessful rounds. We define the group success probabilities as $p_1 = 1 - (1 - p)^{m-1}$ for a reduced committee of size $m-1$, and $p_2 = 1 - (1 - p)^m$ for a full committee of size $m$. Conversely, the solid blue line corresponds to the participation strategy, where miner $i$ contributes its fixed hash power. Its expected payoff per active round is $p_2/m-c_i$.

Figure~\ref{fig:pow-timing-mdp}(b) depicts the \textit{early termination} decision.
At state $S_1$, miner $i$ commits to making at most $k$ attempts, stopping after attempt $k$ if the block has not already been found.
At $S_k$, the miner makes its final attempt and then its participation ends with probability 1.
The expected incremental payoff in this last round is $p_2/m-c_i$, accounting for the success rate when all $m$ miners participate.
The terminal state ends miner $i$'s payoff accounting, even if the other miners continue.

The following proposition compares these two timing deviations with participation until block discovery, conditional on nonnegative expected payoff.

\begin{proposition}
    In the steady state for participants, each agent will have no incentive to delay mining or early terminate if the decision is made at the beginning of the block.\label{prop pow}
\end{proposition}

\begin{proof} 

\textbf{Analysis of Figure~\ref{fig:pow-timing-mdp}(a):} If agent \( i \) chooses \emph{not} to delay mining at the beginning, the expected payoff over the entire process is:
\[
    u_i = \sum_{t=0}^\infty p_2 (1 - p_2)^t \left[\frac{1}{m} - (t + 1)c_i\right] = \frac{1}{m} - \frac{c_i}{p_2}.
\]

However, if the agent chooses to delay mining, then the expected payoff becomes:
\[
u_i' = (1 - p_1)^k \sum_{t=0}^\infty p_2 (1 - p_2)^t \left[\frac{1}{m} - (t + 1)c_i\right] 
= (1 - p_1)^k \cdot \left(\frac{1}{m} - \frac{c_i}{p_2}\right) \leq u_i.
\]

Hence, delaying weakly reduces expected payoff, with a strict reduction when $u_i>0$.

\textbf{Analysis of Figure~\ref{fig:pow-timing-mdp}(b):} Suppose miner $i$ makes at most $k$ attempts and stops after attempt $k$ if the block is still undiscovered.
The expected payoff is:
\[
\begin{aligned}
u_i'' &= \sum_{t=0}^{k-2} p_2 (1 - p_2)^t \left[\frac{1}{m} - (t + 1)c_i\right] 
+ (1 - p_2)^{k-1} \left[p_2 \cdot \frac{1}{m} - k c_i\right] \\
&=\left[1-(1-p_2)^k\right]\left(\frac{1}{m}-\frac{c_i}{p_2}\right)
 \le u_i.
\end{aligned}
\]

For $k=1$, the sum is empty.
The final term includes the costs of all $k$ attempts when the last round is reached.
Since $0<p_1,p_2<1$, both deviations reduce payoff when $u_i>0$ and leave it unchanged when $u_i=0$.
Therefore, an individually rational agent has no profitable delay or early-exit deviation under the stated assumptions.
\qedhere
\end{proof}

With this proposition then we can safely use general utility function for each miner per block in the Equation (\ref{utilityPoW}). The analysis from MDP assumes each agent have the same reward but different cost, it's equivalent to assume the same cost and different reward, as Equation (\ref{utilityPoW}) shows. In other words, the two-stage game can be extended to PoW analysis to find the optimal mining difficulty $p^*$.

\end{document}